\documentclass[11pt, letterpaper]{article}
\usepackage[utf8]{inputenc}
\usepackage[UKenglish]{babel} 

\usepackage[margin=1in]{geometry}

\usepackage{setspace}
\usepackage{authblk} 

\usepackage[bottom]{footmisc}

\usepackage{graphicx} 

\usepackage{xcolor}

\usepackage{amsmath, amsthm, amssymb}

\usepackage{subcaption}

\usepackage{nicefrac}

\usepackage{tcolorbox}
\tcbuselibrary{skins,breakable}
\tcbset{enhanced jigsaw}

\usepackage[compact]{titlesec}

\usepackage{bm}

\usepackage{enumitem}

\newcommand{\Qed}[1]{\hfill\ensuremath{\qed_{\textnormal{~#1}}}\renewcommand{\qedsymbol}{}}

\usepackage{algorithm}
\usepackage[noend]{algpseudocode}

\usepackage[pagebackref=true]{hyperref}
\hypersetup{
    colorlinks=true,
    linkcolor=[rgb]{0.5,0,0},
    filecolor=magenta,   citecolor=[rgb]{0.7,0.3,0},
    urlcolor=[rgb]{0,0,0.5},
    linktocpage=true,
    hypertexnames=false, 
}

\renewcommand*{\backref}[1]{}
\renewcommand*{\backrefalt}[4]{{%
    \ifcase #1 Not cited.%
          \or Cited on page~#2.%
          \else Cited on pages #2.%
    \fi%
    }}

\usepackage[nameinlink, capitalise]{cleveref}

\usepackage{thmtools}
\usepackage{thm-restate}

\usepackage{xspace}

\theoremstyle{plain}
\newtheorem{theorem}{Theorem}
\newtheorem{lemma}[theorem]{Lemma}

\newtheorem{proposition}[lemma]{Proposition}
\newtheorem{problem}{Problem}
\newtheorem{claim}[theorem]{Claim}

\newtheorem{theoremrefinner}{Theorem}
\newenvironment{theoremref}[1]{%
  \renewcommand\thetheoremrefinner{#1}%
  \theoremrefinner
}{\endtheoremrefinner}

\theoremstyle{definition}
\newtheorem{definition}[theorem]{Definition} 

\newtheorem{remark}[theorem]{Remark}
\newtheorem{assumption}[theorem]{Assumption}

\newtheorem*{remark*}{Remark}

\crefname{claim}{claim}{claims}
\crefname{assumption}{assumption}{assumptions}

\title{Space Optimal Vertex Cover in Dynamic Streams}
\author[1]{Kheeran K. Naidu}
\author[2]{Vihan Shah}
\affil[1]{Department of Computer Science, University of Bristol, \texttt{\href{mailto:kheeran.naidu@bristol.ac.uk}{kheeran.naidu@bristol.ac.uk}}}
\affil[2]{Department of Computer Science, Rutgers University, \texttt{\href{mailto:vihan.shah98@rutgers.edu}{vihan.shah98@rutgers.edu}}}
\date{}

\newcommand{\eps}{\ensuremath{\varepsilon}}
\newcommand{\Paren}[1]{\Big(#1\Big)}

\newcommand{\bracket}[1]{\left[#1\right]}
\newcommand{\paren}[1]{\ensuremath{\left(#1\right)}\xspace}
\newcommand{\card}[1]{\left\vert{#1}\right\vert}

\newcommand{\prob}[1]{\Pr\paren{#1}}
\newcommand{\expect}[1]{\Exp\bracket{#1}}
\newcommand{\var}[1]{\textnormal{Var}\bracket{#1}}
\newcommand{\cov}[1]{\textnormal{Cov}\bracket{#1}}

\newcommand{\set}[1]{\ensuremath{\left\{ #1 \right\}}}
\newcommand{\poly}{\mbox{\rm poly}}

\newcommand{\alg}{\ensuremath{\mathcal{A}}\xspace}

\DeclareMathOperator*{\Exp}{\ensuremath{{\mathbb{E}}}}
\DeclareMathOperator*{\Prob}{\ensuremath{\textnormal{Pr}}}
\renewcommand{\Pr}{\Prob}

\newenvironment{tbox}{\begin{tcolorbox}[
		enlarge top by=5pt,
		enlarge bottom by=5pt,
		 breakable,
		 boxsep=0pt,
                  left=4pt,
                  right=4pt,
                  top=10pt,
                  arc=0pt,
                  boxrule=1pt,toprule=1pt,
                  colback=white
                  ]
	}
{\end{tcolorbox}}

\newcommand{\II}{\ensuremath{\mathbb{I}}}

\newcommand{\mireal}[1][]{
  \ifx\relax#1\relax%
    \II(\mione \,; \mitwo)%
  \else%
    \II(\mione \,; \mitwo\mid #1)%
  \fi
}

\renewcommand{\poly}{\mathop{\mathrm{poly}}\nolimits}

\newcommand{\Vopt}{\ensuremath{V^*}\xspace}
\newcommand{\sizeopt}{\textnormal{\ensuremath{\mbox{opt}}}\xspace}

\newtheorem{mdalg}{Algorithm}

\definecolor{ForestGreen}{rgb}{0.1333,0.5451,0.1333}

\newcommand{\vect}{\textnormal{\textbf{vec}}\xspace}

\newcommand{\NES}{\textnormal{\texttt{NE-Sampler}}\xspace}

\newcommand{\snes}{s_{\textnormal{\texttt{NES}}}}

\newcommand{\Measy}{\ensuremath{M_{\textnormal{easy}}}\xspace}
\newcommand{\Geasy}{\ensuremath{{G}_{R}}\xspace}

\newcommand{\ME}[1]{\ensuremath{M_{#1}}\xspace}
\newcommand{\GE}[1]{\ensuremath{{G}_{R}^{#1}}\xspace}

\newcommand{\High}[1]{\ensuremath{\textnormal{High}_{#1}}\xspace}
\newcommand{\Med}[1]{\ensuremath{\textnormal{Med}_{#1}}\xspace}
\newcommand{\Low}[1]{\ensuremath{\textnormal{Low}_{#1}}\xspace}

\newcommand{\degG}[2]{\ensuremath{\textnormal{deg}_{#2}(#1)}\xspace}

\newcommand{\eventmatch}{\ensuremath{\mathcal{E}_{\textnormal{M}}}\xspace}
\newcommand{\eventsparsify}{\ensuremath{\mathcal{E}_{\textnormal{S}}}\xspace}

\newcommand{\com}{\ensuremath{\textnormal{com}}\xspace}
\newcommand{\Com}{\ensuremath{\textnormal{Com}}\xspace}

\newcommand{\NL}{\ensuremath{N\!L}}

\newcommand{\Veasy}{\ensuremath{V(\Measy)}\xspace}
\newcommand{\I}{\ensuremath{\mathcal{I}}\xspace}
\newcommand{\GR}{\ensuremath{G_R}\xspace}

\newcommand{\algms}{\ensuremath{\mathcal{ALG}_\textit{MS}}\xspace}
\newcommand{\algnt}{\ensuremath{\mathcal{ALG}_\textit{NT}}\xspace}

\newcommand{\gvc}{\ensuremath{V_C'}\xspace}

\newcommand{\greedy}{\ensuremath{\textsc{Greedy}}\xspace}

\newcommand{\parms}{\ensuremath{\beta}\xspace}
\newcommand{\mors}{\ensuremath{\textsf{Match-or-Sparsify}_{\parms}}\xspace}
\newcommand{\morsn}{\ensuremath{\textsf{Match-or-Sparsify}_{n}}\xspace}
\newcommand{\morsm}{\ensuremath{\textsf{Match-or-Sparsify}_{\mu}}\xspace}

\newcommand{\mvc}{$\alpha$\textsf{MVC}\xspace}
\newcommand{\mm}{$\alpha$\textsf{MM}\xspace}

\begin{document}
\maketitle
\begin{abstract}
    We optimally resolve the space complexity for the problem of finding an \textsf{$\alpha$-approximate minimum vertex cover} (\mvc) in dynamic graph streams. We give a randomised algorithm for \mvc which uses $O(n^2/\alpha^2)$ bits of space matching Dark and Konrad's lower bound [CCC~2020] up to constant factors.
    By computing a random greedy matching, we identify `easy' instances of the problem which can trivially be solved by returning the entire vertex set. The remaining `hard' instances, then have sparse induced subgraphs which we exploit to get our space savings and solve \mvc.
    
    Achieving this type of optimality result is crucial for providing a complete understanding of a problem, and it has been gaining interest within the dynamic graph streaming community.
    For connectivity, Nelson and Yu [SODA~2019] improved the lower bound showing that $\Omega(n \log^3 n)$ bits of space is necessary while Ahn, Guha, and McGregor [SODA~2012] have shown that $O(n \log^3 n)$ bits is sufficient. 
    For finding an \textsf{$\alpha$-approximate maximum matching}, the upper bound was improved by Assadi and Shah [ITCS~2022] showing that $O(n^2/\alpha^3)$ bits is sufficient while Dark and Konrad [CCC~2020] have shown that $\Omega(n^2/\alpha^3)$ bits is necessary. The space complexity, however, remains unresolved for many other dynamic graph streaming problems where further improvements can still be made.
    \end{abstract}
    
\clearpage

\section{Introduction}
\textit{Graph streaming} is a setting in which a graph is specified by a sequence of edges, typically in arbitrary order.
It is particularly useful for processing massive graphs where having random access to the edges of the graph is either impossible or computationally infeasible.

Research in this area began with \emph{insertion-only streams}, where the stream is made up of a sequence of edge insertions only. 
In their seminal work, Feigenbaum, Kannan, McGregor, Suri, and Zhang \cite{fkmsz04} showed that for many problems including minimum spanning tree, connectivity, and bipartiteness, $\Omega(n)$ bits of space is necessary and $O(n \log n)$ bits is sufficient for any $n$-vertex graph. 
This logarithmic gap was often overlooked and deemed not important when proving optimality for graph problems, but it left unresolved the question of whether the logarithmic factor was required for simply storing edges or if other techniques could remove it.
About a decade later, Sun and Woodruff \cite{sun2015tight} showed that the logarithmic factor was indeed necessary by improving the lower bounds to $\Omega(n \log n)$ bits, asymptotically matching the upper bounds up to constant factors.


\textit{Dynamic graph streams}, which allow for sequences of both edge insertions and deletions, prove to be more difficult. 
Edges that arrive in the stream are not necessarily in the final graph as they may later be deleted. 
In fact, it is well-known in the community that it is impossible to deterministically return a single edge of a dense graph without storing all of its edges.
As a result, almost all dynamic graph streaming algorithms rely on counters which use $O(\log n)$ bits of space or they rely on $L_0$-sampling which optimally uses $\Theta(\log^3 n)$ bits of space\footnote{
This optimal space bound applies when the probability of success is at least $1 -\frac{1}{\poly(n)}$.
} \cite{jst11,KapralovNPWWY17}.
In essence, counters are used to solve the problem of determining whether an edge is present in an edge induced subgraph \cite{dk20} (see also \cite{ccehmmv16}), whereas $L_0$-sampling also returns the identity of a uniform random edge if one is present \cite{agm12,agm12b,k15,ccehmmv16,akly16, assadi2019sublinear,konrad2021frequent,kk22,as22}.
A notable exception includes spectral sparsification \cite{klmms14,kmmmnst20} which relies on $L_2$-heavy-hitters (non-uniform sampling).


Resolving the space complexity up to constant factors for dynamic graph streaming problems has continued to be an elusive task. Ahn, Guha, and McGregor \cite{agm12} gave an algorithm for connectivity using $O(n \log^3 n)$ bits of space, and for several years, the best known lower bound was the insertion-only bound of $\Omega(n \log n)$ bits \cite{sun2015tight}. However, in 2019, Nelson and Yu \cite{nelson2019optimal} improved the lower bound to $\Omega(n \log^3 n)$ bits in the dynamic graph streaming setting. 
To the best of our knowledge, this is the only problem in this setting which has space bounds that prove the necessity of the $\Theta(\log^3 n)$ overhead of randomly sampling an edge (using $L_0$-sampling).
The approximate minimum cut problem which has a $\Omega(n \log^3 n)$ bit lower bound \cite{nelson2019optimal} (and a $O(n \log^4 n)$ bit upper bound \cite{agm12}) similarly shows that logarithmic factors are necessary.
A perhaps more surprising result was the recent progress on \textsf{$\alpha$-approximate maximum matching} (\mm).
The lower bound of $\Omega(\nicefrac{n^2}{\alpha^3})$ bits \cite{dk20} (see also \cite{akly16}) and the previous upper bound of $O(\nicefrac{n^2}{\alpha^3} \cdot \log^4 n)$ bits \cite{akly16,ccehmmv16} seem to indicate that the logarithmic overhead of sampling an edge is required. However, Assadi and Shah \cite{as22} improved the upper bound to $O(\nicefrac{n^2}{\alpha^3})$ bits showing that this is not the case.
On the other hand, for problems such as vertex cover \cite{dk20}, dominating set \cite{kk22}, and spectral 
sparsification \cite{kmmmnst20}, their space bounds have a gap of logarithmic factors, and therefore further improvements can still be made.

\subparagraph{Our Results.}
In this work, we optimally resolve the space complexity up to constant factors for the problem of finding an \textsf{$\alpha$-approximate minimum vertex cover} (\mvc) in a dynamic graph stream. In particular, we improve the upper bound to $O(\nicefrac{n^2}{\alpha^2})$ bits, matching the $\Omega(\nicefrac{n^2}{\alpha^2})$ bits lower bound \cite{dk20} and showing that the logarithmic overhead is not required. Our main result is the following:

\begin{theorem}\label{thm:vcalg}
There exists a randomised dynamic graph streaming algorithm for \mvc that succeeds with high probability and uses $O(\nicefrac{n^2}{\alpha^2})$ bits of space for any $\alpha \leq n^{1 - \delta}$ where $\delta > 0$.
\end{theorem}

\subparagraph{Previous Work.} 
It has been shown by Dark and Konrad \cite{dk20} that $\Omega(\nicefrac{n^2}{\alpha^2})$ bits is necessary for \mvc. They also gave a simple deterministic algorithm which uses $O(\nicefrac{n^2}{\alpha^2} \cdot \log \alpha)$ bits of space, matching the lower bound up to logarithmic factors. 
Their algorithm arbitrarily partitions the vertex set into $\nicefrac{n}{\alpha}$ groups of size $\alpha$ and uses counters, which introduce the logarithmic overhead, to maintain the number of edges between each of the $\Theta(\nicefrac{n^2}{\alpha^2})$ pairs of vertex groups. 
The solution follows by computing a group-level minimum vertex cover, and then returning the vertices of the covering groups.



\subparagraph{Main Techniques.}
We improve the approach of Dark and Konrad \cite{dk20} by additionally computing a supporting random \greedy matching and randomly partitioning the vertex set into $\nicefrac{n}{\alpha}$ groups, effectively using randomisation to reduce the space required.
The random \greedy matching returned is either large enough to imply a trivial solution for \mvc (`easy' case) or implies sparseness properties of the residual subgraph induced by the unmatched vertices (`hard' case).
To solve the `hard' cases, we use the sparseness properties and the random partitioning to argue that there are only $O(1)$ many edges between each pair of vertex groups in the residual subgraph. Therefore, storing edge counters for each of the $\Theta(\nicefrac{n^2}{\alpha^2})$ many pairs, as done by Dark and Konrad \cite{dk20}, now requires only $O(\nicefrac{n^2}{\alpha^2})$ bits of space in total.

\subparagraph{Sampling Strategies.}
The sparseness properties (of the residual subgraph) implied are reliant on the method of randomly sampling edges from the graph. 
Uniformly sampling from the edge set only implies sparseness properties sufficient for a small range of $\alpha$ since it is skewed to sampling high degree vertices.
On the other hand, non-uniform sampling -- sampling from the neighbourhood of a random set of vertices, coined {\em neighbourhood edge sampling} by Assadi and Shah \cite{as22} -- is less biased towards high degree vertices and implies the necessary sparseness properties for the full range of $\alpha$.
Indeed, Assadi and Shah \cite{as22} also use the approach of computing a \greedy matching on non-uniformly sampled edges to identify the `easy' and `hard' instances of \mm.
However, for \mvc, our `easy' and `hard' instances differ from those of \mm, so we require different guarantees. Furthermore, we use different techniques for solving the `hard' instances.


\subparagraph{Further Related Work.}
Resolving the space complexity up to constant factors has also been achieved for non-graph problems in the general data streaming setting.
For instance, Braverman, Katzman, Seidell, and Vorsanger \cite{braverman2014optimal} gave an upper bound for finding a constant factor approximation to the $k$-th frequency moment in constantly many passes that matches the lower bound of Woodruff and Zhang \cite{woodruff2012tight}.
Price and Woodruff \cite{price2013lower} showed a lower bound for any adaptive sparse recovery scheme that matches the upper bound of Indyk, Price, and Woodruff \cite{indyk2011power}.
Graph problems in other streaming settings have also been studied. For example, the settings which allow multiple passes over the stream \cite{kmm12,k18,ack19a,aksy20,kn21,a22,ajjst22}, have a random arrival order \cite{kmm12,abbms19,b20}, or have highly structured deletions via a sliding window \cite{cms13,cs14,bbm21} have been considered.
See the work by McGregor \cite{m14} for an excellent survey on graph streaming algorithms.

\paragraph{Outline.}
We begin in \Cref{sec:prelim} with some important notation and tools which we will later use.
In \Cref{sec:randomgreedy}, we discuss the guarantees required from a random \greedy matching for \mvc.
In \Cref{sec:main-result}, we present and analyse our algorithm that proves \Cref{thm:vcalg}.
Then, we conclude in \Cref{sec:concl}.

\section{Preliminaries} \label{sec:prelim}

For any $n$-vertex graph $G=(V,E)$, let $\mu(G)$ be the size of the maximum matching of the graph, let $\Vopt(G)$ be a minimum vertex cover, and let $\sizeopt(G)$ be its size. 
We will simply use $\mu$, \Vopt or \sizeopt if the graph is clear from context. 
For any subset of edges $F \subseteq E$, we denote the set of their endpoints by $V(F)$. For any subgraph $H$ of $G$ and vertex $v \in V$, we use $N_H(v)$ to denote the neighbourhood of $v$ in $H$ .

The graph $G$ may be specified as a dynamic graph stream\footnote{
A dynamic graph stream is a special case of the strict turnstile data streaming model \cite{m05} where we consider only bit-vectors which represent the edges of a graph.
} 
$\sigma = (\sigma_1, \sigma_2, ..., \sigma_N)$ such that $\sigma_j = (i_j, \Delta_j)$ where $i_j \in [m]$ for $m = \binom{n}{2}$ and $\Delta_j \in \set{1,-1}$ (insertions or deletions).
Note that edges may only be deleted if they have previously been inserted.
Additionally, the stream must produce a vector $\textbf{vec}(E) \in \set{0,1}^m$ that defines the edge set $E$, i.e., the $i^{th}$ entry of the vector indicates the presence of the edge indexed by $i \in [m]$.

In our work, we will rely on limited independence hash functions to reduce the space complexity of our algorithm.
Roughly speaking, a hash function sampled from a family of $k$-wise independent hash functions behaves like a totally random function when considering at most $k$ elements.
For simplicity, when we mention a $k$-wise independent hash function, we will mean a hash function sampled from a family of $k$-wise independent hash functions.
We use the following standard result for $k$-wise independent hash functions.
\begin{proposition}[\cite{MotwaniR95}]\label{prop:k-wise}
	For all integers $n,m,k \geq 2$, there is a family of $k$-wise independent hash functions $\mathcal{H} = \set{h: [n] \rightarrow [m]}$ such that sampling and storing a function $h \in \mathcal{H}$ takes $O(k \cdot (\log n + \log m))$ bits of space.
\end{proposition}

We shall also use the following concentration result on an extension of Chernoff-Hoeffding bounds for $k$-wise independent hash functions. 
\begin{proposition}[\cite{SchmidtSS95}]\label{prop:chernoff-limited}
	Suppose $h$ is a $k$-wise independent hash function and $X_1,\ldots,X_m$ are $m$ random variables in $\set{0,1}$ where $X_i = 1$ iff $h(i) = 1$. 
   Let $X := \sum_{i=1}^{m} X_i$. 
   Then, for any $\eps > 0$, 
   \[
	   \prob{\card{X - \expect{X}} \geq \eps \cdot \expect{X}} \leq \exp\paren{-\min\set{\frac{k}{2},\frac{\eps^2}{4+2\eps} \cdot \expect{X}}}. 
   \] 
\end{proposition}

Finally, we will use the following sketching tool for dynamic graph streams to test the size of the neighbourhood of a subset of vertices.
\begin{proposition}[{\cite{as22}}] \label{prop:Ntester}
Let $a \geq b \geq 2$ be known integers. 
Consider a $n$-vertex graph $G = (V,E)$ specified in a dynamic stream and let $S \subseteq V$ be a known set. 
Then, given a set $T \subseteq V$ of size at most $a$ at the end of the stream, there exists a randomised algorithm that returns ``Yes'' if $|N_G(S) \backslash T| \geq b$ or ``No'' if $|N_G(S) \backslash T| \leq \frac{1}{2} \cdot b$, uses $O(\frac{a}{b} \cdot \log^3 n)$ bits of space, and succeeds with probability at least $1 - n^{-3}$. We denote one such algorithm as $\algnt(S; a,b)$.
\end{proposition}

\section{Sampling Strategies for Random Greedy Matchings }\label{sec:randomgreedy}
In this section, we discuss and present the tool that we use to either find a large matching or show that the residual subgraph induced by the unmatched vertices is sparse. 

This approach was also used in Assadi and Shah's recent work for \mm \cite{as22} to identify `easy' and `hard' instances of the problem.
For \mvc, the `easy' case is finding a large enough matching to imply that we can trivially return the entire vertex set to solve the problem.
The `hard' case is when we get a sparse residual subgraph, which is where our main savings in space come from. 
Identifying these cases can be accomplished by computing a \greedy matching on randomly sampled edges of the graph.

Uniformly sampling as many edges as possible from a $n$-vertex graph (using $L_0$-sampling) without exceeding $O(\nicefrac{n^2}{\alpha^2})$ bits of space followed by computing a \greedy matching implies sparseness properties based on an already known maximum degree bound of the residual subgraph induced by the unmatched vertices \cite{ahn2015correlation,konrad2018mis,gkms18ext}. 
Intuitively, uniform sampling is skewed towards sampling edges incident to high degree vertices. Hence, a \greedy matching either matches these high degree vertices or matches many of its neighbours (decreasing their residual degree), and regardless of the size of the matching found, this gives a $\poly(\alpha)$ max degree bound in the residual graph.
Furthermore, we can show that this also bounds the average degree (even when a small matching is found) since a worst-case instance\footnote{
Consider a graph with a large clique on $\Theta(\nicefrac{n}{\alpha})$ vertices where most the edges are sampled from, and many smaller cliques which assert the guaranteed max degree bounds.
}
practically has all vertices in the residual subgraph with max degree.
This degree bound, however, is only sufficient for solving \mvc for any $\alpha \ll n^{\frac{1}{3.5}}$.

Non-uniformly sampling the edges using neighbourhood edge sampling followed by $\greedy$, as done by Assadi and Shah \cite{as22}, proves to give better sparseness properties, and thus a better average degree bound\footnote{
Having an average degree bound is more difficult to work with, but in this case, the bound on the average degree is much smaller than the bound on the max degree in the uniform case.
}.
The benefit of neighbourhood edge sampling is that it biases away from sampling high degree vertices. Furthermore, when a small \greedy matching is found, the implication is that the residual subgraph is sparse.
Therefore, the average degree bound is sufficient for solving the `hard' case of \mvc for the full range of $\alpha$.

As previously mentioned, Assadi and Shah's algorithm called Match-or-Sparsify \cite{as22}, does exactly this, although its guarantees are not sufficient for our purposes. Hence, we first discuss their algorithm, and then explain the alterations we make.

\subparagraph{Match-or-Sparsify.} 
For some parameter $\parms \leq n$, Assadi and Shah's \mors algorithm non-uniformly samples edges using space $O(\nicefrac{\parms^2}{\alpha^3})$ bits, and then computes a \greedy matching from them.
They give an intricate analysis to show that their algorithm either finds a large matching of size at least $\nicefrac{\parms}{8\alpha}$ or implies that the residual subgraph has at most $20 \cdot \parms \cdot \log^4 n$ edges \cite[Lemma 16]{as22}.
Unlike uniform sampling, the residual properties (sufficiently) only hold when the matching is small -- a key property exploited in their analysis.
Additionally, in order for the guarantees to hold, they rely on the assumption that $\parms \geq \alpha^2 \cdot n^\delta$. 
Informally, when $\parms$ is set as the size of the maximum matching $\mu$, \morsm finds a large matching in `easy' graph cases and a sparse residual subgraph in the `hard' graph cases. However, $\mu$ is not known, so they find a setting of $\parms$ close to $\mu$ by running \mors in parallel with $\parms$ as all powers of 2 between $1$ and $n$. 



\subparagraph{Our Alterations.}
The first thing to note is that the `easy' and `hard' instances for \mm and \mvc are not the same. 
Consider \mors when a large \greedy matching is found. 
Since at least one endpoint of each matching edge must be in a vertex cover, it implies that $\sizeopt \geq \frac{\parms}{8\alpha}$.
However, returning a solution to \mvc at this stage can only be of size at most $\Theta(\parms)$, which would not be a trivial solution (the entire vertex set) with $\parms \ll n$.
Furthermore, we have no guaranteed sparseness properties since the matching found is large. Hence, instead of needing $\beta \approx \mu$, which requires $\log n$ many runs to find, we only need a single run of \morsn (with the parameter $\parms$ fixed as $n$). 
Secondly, their assumption that $\parms \geq \alpha^2 \cdot n^\delta$ implies that $\alpha \leq n^{\frac{1-\delta}{2}}$, but we require it to hold for any $\alpha \leq n^{1-\delta}$. 
Since we have an additional $\alpha$ factor of space (see \cite{dk20}), we can increase the number of non-uniformly sampled edges to use $O(\nicefrac{n^2}{\alpha^2})$ bits instead, which allows us to remove the assumption.
Finally, the increase in the number of samples also allows us to increase the sparseness guarantees of the residual subgraph by an $\alpha$ factor.
Therefore, this altered \morsn algorithm, denoted by \algms, gives us the following lemma (full proof given in \Cref{app:m-or-s} for completeness).

\begin{restatable}{lemma}{ms}\label{lem:ms}
There is a linear sketch for dynamic graph streams that, given any graph $G=(V,E)$ specified via $\vect(E)$, uses $O(\nicefrac{n^2}{\alpha^2})$ bits of space and with high probability outputs a matching $\Measy$ 
that satisfies \underline{at least one} of the following conditions for any $\alpha \leq n^{1-\delta}$ and $\delta > 0$:  
\begin{itemize}
\item \textbf{Match-case:} The matching $\Measy$ has at least $\frac{n}{8\alpha}$ edges; 
\item \textbf{Sparsify-case:} The induced subgraph of $G$ on vertices not matched by $\Measy$, denoted by $\Geasy$,  has at most $20 \cdot \frac{n}{\alpha} \cdot \log^4 n$ edges.
\end{itemize} 
\end{restatable}

\section{Main Result}\label{sec:main-result}

In this section, we give a dynamic graph streaming algorithm for \mvc for any $n$-vertex graph which implies our main result:

\begin{theoremref}{\ref*{thm:vcalg}} \label{thm:vcalg-full}
There exists a randomised dynamic graph streaming algorithm for \mvc that succeeds with high probability and uses $O(\nicefrac{n^2}{\alpha^2})$ bits of space for any $\alpha \leq n^{1 - \delta}$ where $\delta > 0$.
\end{theoremref}

Before proceeding, we give the following standard assumption (with reason) which simplifies what we need to prove.

\begin{assumption}\label{asm:optapprox}
A randomised dynamic graph streaming $\Theta(\alpha)$-approximation algorithm that uses $O(\nicefrac{n^2}{\alpha^2})$ bits of space and succeeds on graphs where
$\sizeopt \geq \frac{n}{\alpha \cdot \log^2 n}$
is sufficient to prove \Cref{thm:vcalg-full}.
\end{assumption}

\begin{proof}[Reason]
Let $\alg$ be an algorithm that returns a $(c \cdot \alpha)$-approximation using $O(\nicefrac{n^2}{\alpha^2})$ bits of space. 
Run $\alg$ with parameter $\nicefrac{\alpha}{c}$ to get an $\alpha$-approximation which similarly uses $O(\nicefrac{n^2}{\alpha^2})$ bits.

Then, since we can run $\Theta(1)$ many algorithms which use $O(\nicefrac{n^2}{\alpha^2})$ bits of space in parallel without asymptotically increasing the space, we run an additional algorithm which detects and outputs a solution for graphs with small \sizeopt.

\subparagraph{Algorithm for small $\sizeopt$.}
We use the well-known algorithm for finding an exact minimum vertex cover in dynamic graph streams with probability at least $1 - \frac{1}{\poly(k)}$ given the promise that $\sizeopt \leq k$ with $k=\frac{n}{\alpha \cdot \log^2 n} \geq n^{\nicefrac{\delta}{2}}$ \cite{ccehmmv16}. Note that the $\poly(k)$ can be made a function of $\delta$ to get a success probability of at least $1 - {k^{\nicefrac{-20}{\delta}}} \geq 1 - n^{-10}$.
If $\sizeopt < k$, then we get an optimal solution; otherwise, we get a set of vertices of size $k$ which are not necessarily a solution.
Thus, we can detect this case by the size of the returned vertex cover being smaller than $k$.
The space taken by the algorithm is $O(k^2 \cdot \log^4 n) = O(\nicefrac{n^2}{\alpha^2})$ bits and it works for all $\alpha = \omega(1)$ 
(for $\alpha=\Theta(1)$ we can store the entire graph).
\end{proof}


\begin{algorithm*}[ht]
\caption{Optimal Dynamic Vertex Cover} 
\label{alg:optvc}
\textbf{Input:} A dynamic graph stream $\sigma$ for a $n$-vertex graph $G =(V,E)$, a small 
constant $\delta > 0$, and a positive integer $\alpha \leq n^{1-\delta}$ 

\textbf{Output:} A vertex cover $V_C$ of $G$
\vspace{2mm} 
\\ 
\textbf{Pre-processing:}
\begin{algorithmic}[1]
\State Initialise $\mathbf{M}$ to be an instance of \algms (\Cref{lem:ms})
\State Randomly partition $V$ into groups $V_1, V_2, ..., V_{\frac{n}{\alpha}}$ having size in 
$[\alpha/2, 2\alpha]$ \label{optvc:parti}
\State For each group $V_i$, initialise $\mathbf{N}_i$ to be an instance of $\algnt(V_i;a,b)$ (\Cref{prop:Ntester}) with $a = \nicefrac{n}{\alpha}$ and $b = n^{\nicefrac{\delta}{2}}$
\State Set $c = \nicefrac{15}{\delta}$
\label{optvc:countend1}
\end{algorithmic}
\vspace{2mm}
\textbf{Processing the stream:}
\begin{algorithmic}[1]
\setcounter{ALG@line}{\getrefnumber{optvc:countend1}}
\State Update $\mathbf{M}$ and each $\mathbf{N}_i$ using $\sigma$
\State For every pair of groups $V_i$ and $V_j$, store a counter $C_{i,j}$ for the number of edges between them modulo $c$
\State For every group $V_i$, store a counter $C_i$ for the number of internal edges
\label{optvc:countend2}
\end{algorithmic}
\vspace{2mm}
\textbf{Post-processing:}
\begin{algorithmic}[1]
\setcounter{ALG@line}{\getrefnumber{optvc:countend2}}
\State Let $\Measy$ be the matching returned by $\textbf{M}$
\If{$\Measy$ has at least $\frac{n}{8 \cdot \alpha}$ edges} \textbf{return} $V$ \label{optvc:match-case} \EndIf 
\State Let $V_C$ be the union of all groups $V_i$ containing a vertex of $\Measy$ or with $C_i > 0$ \label{optvc:simple}
\State Add to $V_C$ all remaining vertex groups $V_i$ where $\mathbf{N}_i$ returns ``Yes'' when $T = \Veasy$ \label{optvc:residual}
\State Consider the multi-graph $G'$ obtained by contracting the vertices of each remaining vertex group $V_i$ into a single vertex $v_i$ where $C_{i,j}$ represents the multiplicity modulo $c$ of each edge $(v_i,v_j)$ in $G'$
\State Greedily compute a vertex cover $\gvc$ of $G'$
\State For all $v_i \in \gvc$, add vertex group $V_i$ to $V_C$ 
\State \textbf{return} $V_C$
\end{algorithmic}
\end{algorithm*}

\subparagraph{Algorithm Description.} 
Let $G = (V,E)$ be specified by a dynamic graph stream, $\delta > 0$,
and $\alpha \leq n^{1 - \delta}$ be the inputs to \Cref{alg:optvc}. 
The algorithm, in its pre-processing step, partitions $V$ into $\nicefrac{n}{\alpha}$ groups using a $(10 \cdot \log n)$-wise independent hash function (when the space allows, i.e., for small $\alpha$, we do this using a uniform random permutation instead), and we later show that all their sizes lie between $\alpha/2$ and $2\alpha$ with high probability. 
During the stream, it maintains counters modulo some constant for the number of edges between each pair of 
groups and (standard) counters for the number of internal edges of each group. 
In parallel, it computes a random matching $\Measy$ using an instance of \algms (\Cref{lem:ms}) and maintains residual
neighbourhood size testers for each vertex group using instances of \algnt (\Cref{prop:Ntester}). 
In the post-processing step, if the matching is of size at least $\frac{n}{8 \cdot \alpha}$, then the entire 
vertex set is returned. 
Otherwise, the vertex groups containing any vertex of the matching or any internal edges are entirely picked in the solution -- we call these \textit{simple vertex groups}. 
Next, the remaining vertex groups $V_i$ whose residual neighbourhood is large, $|N_{\GR}(V_i)| = |N_G(V_i) \backslash 
\Veasy| \geq n^{\delta/2}$ where $\GR = G[V \backslash \Veasy]$, are added to the solution -- we call these \textit{residual vertex groups}. 
Finally, among the leftover \textit{clean vertex groups}, the algorithm uses the counters modulo 
some constant to perform a group-level vertex cover, and then further adds the covering groups to the solution 
before returning it.

\begin{definition}[Simple Vertex Groups]\label{defn:simple}
We say that a vertex group $V_i$ is \textit{simple} if any of its vertices are matched by $\Measy$ or 
it has at least one internal edge, i.e., $|V_i \cap \Veasy| > 0$ or $C_i > 0$.
\end{definition}

\begin{definition}[Residual Vertex Groups]\label{defn:residual}
We say that a vertex group $V_i$ is \textit{residual} if it is not simple and has a large residual 
neighbourhood, i.e., $|N_{\GR}(V_i)| \geq n^{\delta/2}$.
\end{definition} 

\begin{definition}[Clean Vertex Groups]\label{defn:clean}
We say that a vertex group $V_i$ is \textit{clean} if it is not simple or residual, i.e., $|V_i \cap \Veasy| = 0$, $C_i = 0$ and $|N_{\GR}(V_i)| < n^{\delta/2}$.
\end{definition}

Note that throughout the subsequent analysis of \Cref{alg:optvc}, all results succeed with high probability. 
Hence, at any point, we can do a simple union bound to show that they all hold with high probability.
As such, we condition on this event here to avoid explicitly doing so during the analysis.


Let $G$ be the input graph of the algorithm. We begin the analysis with the following observation: If $G$ contains a matching of size at least 
$\frac{n}{8 \cdot \alpha}$, then $V$ is a valid $(8 \cdot \alpha)$-approximation of a minimum vertex cover $\Vopt$ since at least one endpoint of a matching edge must be in a valid vertex cover.
Therefore, if the condition of \Cref{optvc:match-case} is satisfied, the 
algorithm terminates and the solution is a valid $\Theta(\alpha)$-approximation (`easy' graph instances). 
Otherwise, the 
algorithm progresses with $|\Measy| < \frac{n}{8 \cdot \alpha}$, i.e., the sparsify-case of \Cref{lem:ms} (`hard' graph instances). 
This implies that the residual subgraph $\GR$ is sparse with at most $20 \cdot \frac{n}{{\alpha}} \cdot \log^4 n$ many edges.
As such, we need to prove that we also get a $\Theta(\alpha)$-approximation in the sparsify-case. 

We highlight here that the algorithm adds vertex groups to the solution for various reasons, which are determined by whether it is a simple, residual, or clean vertex group (see \Cref{defn:simple,defn:residual,defn:clean}). 
Hence, we proceed with the analysis of the sparsify-case by considering these different types of vertex groups separately.

\subparagraph{Simple Vertex Groups.} 
Let $\I_s$ be the index set of the simple vertex groups. We argue that there are not too many of these, so we can add all of them to the solution.
\begin{claim} \label{clm:simple}
The number of simple vertex groups $|\I_s|$ is at most $2 \cdot \sizeopt(G)$.
\end{claim}
\begin{proof}
Each edge of the matching \Measy can cause up to two vertex groups to be classified as simple; however, they must have at least one vertex of \Vopt since at least one endpoint of every matching edge must be in \Vopt. 
Therefore, for every two groups classified as simple in this way, there is at least one vertex of \Vopt in their union. 
On the other hand, a group could also be classified as simple if it contains an internal edge, where one of its endpoints must be in \Vopt. 
Hence, for each group classified as simple in this way, there is at least one vertex of \Vopt in it.
Then, it follows that the number of simple vertex groups must be at most $2 \cdot\card{\Vopt} =  2\cdot\sizeopt$.
\end{proof}

\subparagraph{Residual Vertex Groups.} 
Let $\I_r$ be the index set of the residual vertex groups. 
Recall that any residual vertex group must have at least $n^{\nicefrac{\delta}{2}}$ many 
residual neighbours. 
We note, however, that due to the guarantees of the neighbourhood size tester algorithm $\algnt$ (see \Cref{prop:Ntester}), there are some misclassifications, so some residual vertex groups are also of size between $\frac{1}{2} \cdot n^{\nicefrac{\delta}{2}}$ and $n^{\nicefrac{\delta}{2}}$.
This will not be an issue, and moving forward, when we mention residual vertex groups, we assume that this includes the misclassifications. 
Now, we argue that there are not too many residual vertex groups, so we can add them all to the solution.

\begin{claim} \label{clm:residual}
The number of residual vertex groups $|\I_r|$ is at most $\sizeopt(G)$ with high probability.
\end{claim}
\begin{proof}
We have that $\card{\Veasy}$ is at most $\frac{n}{4\cdot\alpha}$ and \GR has at most $20 \cdot \frac{n}{{\alpha}} \cdot \log^4 n$ many edges. 
As such, \GR has $n - \card{\Veasy} \geq \frac{n}{2}$ vertices, and the average degree of a vertex in $\GR$ is at most $20 \cdot \frac{n}{{\alpha}} \cdot \log^4 n \cdot \frac{2}{n}=\frac{40 \log^4 n}{\alpha}$.
Since each non-simple vertex group $V_i$ is fully contained in \GR and has at most $2\alpha$ vertices, we have that $\expect{\card{N_{\GR}(V_i)}} \leq \frac{40 \log^4 n}{\alpha} \cdot 2 \alpha = 80 \log^4 n$. Then, it follows by Markov's inequality that
\begin{align}
\begin{split}
\prob{\textrm{$V_i$ is residual} \mid \textrm{$V_i$ is non-simple}}
&\leq \prob{|N_{\GR}(V_i)| \geq \frac{1}{2} \cdot n^{\nicefrac{\delta}{2}}} \\
&\leq \frac{2 \cdot 80 \log^4 n}{n^{\nicefrac{\delta}{2}}} 
\leq \frac{\log^5 n}{n^{\nicefrac{\delta}{2}}}.
\end{split} \label{eq:residualN}
\end{align}

Let $X_i$ be the indicator random variable that a non-simple vertex 
group $V_i$ is a residual vertex group, then $R = \sum_{i \in [\frac{n}{\alpha}] \backslash \I_s} 
X_i$ is the number of residual vertex groups. 
By \Cref{eq:residualN}, we have the following:
\begin{align*}
\mathbb{E}[R] 
= \sum_{i \in [\frac{n}{\alpha}] \backslash \I_s} \prob{X_i} 
\leq \sum_{i \in [\frac{n}{\alpha}]} \frac{\log^5 n}{n^{\nicefrac{\delta}{2}}} 
= \frac{n \cdot \log^5 n}{\alpha \cdot n^{\nicefrac{\delta}{2}}}.
\end{align*}
Finally, since $\sizeopt(G) \geq \frac{n}{\alpha \cdot \log^2 n}$ (\Cref{asm:optapprox}), a further application of Markov's inequality implies the result:
\begin{align*}
\prob{|\I_r| > \sizeopt} \leq \prob{R > \frac{n}{\alpha \cdot \log^2 n}}
\leq \frac{n \cdot \log^5 n}{\alpha \cdot n^{\nicefrac{\delta}{2}}} \cdot \frac{\alpha \log^2 n}{n} 
\leq  n^{-\nicefrac{\delta}{4}}.
\end{align*}

Note that we can easily increase the success probability by running the algorithm in parallel $40/\delta$ times and detecting failures when the number of residual groups is more than $\nicefrac{n}{\alpha \cdot \log^2 n}$. 
Then, with probability at least $1-n^{-10}$, one of the runs will succeed.
This only increases the space of the algorithm by a constant factor since $40/\delta = \Theta(1)$.
\end{proof}

\subparagraph{Clean Vertex Groups.} 
Let $\I_c$ be the index set of the clean vertex groups and let $\I_c^+$ be the ones added to the solution, which also corresponds to the group-level vertex cover $\gvc$ in \Cref{alg:optvc}.

Before analysing the group-level vertex cover, we note that the relevant counters are stored modulo $c$. 
This means that if the number of edges between clean vertex groups is some multiple of $c$, the corresponding counter would be $0$ and the group-level vertex cover would be incorrect. 
Hence, we want the number of edges between clean vertex groups to be less than $c$ with high probability.

\begin{claim} \label{clm:groupsedges}
For all pairs of clean vertex groups $V_i$ and $V_j$, with high probability, \[|N_G(V_i) \cap V_j| < c.\]
\end{claim}
\begin{proof}
We prove a slightly generalised statement which implies what we need.
We show that there are less than $c$ edges of $\GR$ between any clean vertex group $V_i$ and any other vertex group $V_j$. This implies what we need since, by definition, all edges between clean vertex groups are in $\GR$.

Consider the random partitioning of $V$ using an at least $(3\cdot c)$-wise independent hash function (the algorithm uses $(10 \cdot \log n)$-wise independence).
A residual neighbour of the clean vertex group $v \in N_{\GR}(V_i)$ uniformly belongs to any of the other vertex groups.
Since there are $\frac{n}{\alpha} - 1$ of these (including $V_j$, but not including $V_i$), the probability that $v \in V_j$ is at most $\frac{2\alpha}{n}$.

Now, since clean vertex groups are non-residual, $|N_{\GR}(V_i)| \leq n^{\nicefrac{\delta}{2}}$, and for a fixed $V_i$ and $V_j$, we have that
\begin{align*}
    \prob{ |N_{\GR}(V_i) \cap V_j| \geq c}
    & \leq \binom{n^{\nicefrac{\delta}{2}}}{c} \cdot \left(\frac{2\alpha}{n} \right)^c 
    \leq \left(\frac{2\alpha}{n^{1-\nicefrac{\delta}{2}}} \right)^{\nicefrac{15}{\delta}}
    \leq n^{-7}
\end{align*}
where we have used $\alpha \leq n^{1 - \delta}$ and $c = \nicefrac{15}{\delta}$ in the final inequalities. 
Then, the result holds with probability at least $1 - n^{-5}$ by a union bound over all pairs of vertex groups.

Note that for small $\alpha$ we will partition $V$ into groups of size exactly $\alpha$ with a uniform random permutation due to concentration and space reasons (see \Cref{clm:partVspace}), but the above arguments also hold in this case.
\end{proof}

With \Cref{clm:groupsedges}, we can assume that all the counters between clean vertex groups count exactly the number of edges with high probability, that is, the modulo has no effect on the correctness of the algorithm. 
Thus, the setting is now identical to that of Dark and Konrad's algorithm \cite{dk20}, and we follow a similar argument as they did to analyse the group-level vertex cover and the corresponding subset of clean vertex groups added.

\begin{claim} \label{clm:clean}
The number of clean vertex groups added $\card{\I_c^+}$ is at most $2\cdot\sizeopt(G)$.
\end{claim}
\begin{proof}
Consider the subgraph $H = G[\cup_{i \in \I_c} V_i]$ induced by the clean vertex groups. 
Observe that since $H$ is an induced subgraph of $G$, $\sizeopt(H) \leq \sizeopt(G)$.
Then, since the vertex contractions to obtain the multi-graph $G'$ from $H$ cannot increase the size of its minimum vertex cover, we have that $\sizeopt(G') \leq \sizeopt(H)$. 
Finally, since we greedily compute the group-level vertex cover $\gvc$, it is a $2$-approximation and we have that $\card{\I_c^+} =\card{\gvc} \leq 2 \cdot \sizeopt(G') \leq 2 \cdot \sizeopt(G)$.
\end{proof}

By combining the analysis of the simple, residual, and clean vertex groups, we prove the approximation factor of the algorithm.
\begin{lemma}\label{lem:vcoptapprox}
\Cref{alg:optvc} returns a valid $\Theta(\alpha)$-approximation of a minimum vertex cover for any input graph $G$ with $\sizeopt \geq \frac{n}{\alpha \cdot \log^2 n}$.
\end{lemma}
\begin{proof}
We first show that the solution $V_C$ is indeed a valid vertex cover, then we prove that it is a $\Theta(\alpha)$-approximation.

\subparagraph{Validity.} 
For the sake of finding a contradiction, let $e \in E$ be an edge which is not covered by $V_C$. 
Observe that any non-clean vertex group $V_i$ is added to $V_C$; thus, all edges with at least one endpoint in any of these vertex groups are covered. 
So, we have that $e$ must be in $G[\cup_{i \in \I_c} V_i]$, the subgraph induced by the clean vertex groups.

Let $i,j \in \I_c$ be such that $e$ has endpoints in the clean vertex groups $V_i$ and $V_j$, implying that there is an edge between their corresponding contracted vertices $v_i$ and $v_j$ in the multi-graph $G'$. 
It follows that one of $v_i$ or $v_j$ must be in the computed group-level vertex cover $\gvc$, so all vertices of either $V_i$ or $V_j$, including at least one endpoint of $e$, are added to $V_C$. 
However, this means that $e$ is covered by $V_C$, a contradiction.

\subparagraph{Approximation.} 
Observe that the solution $V_C$ is comprised of a (disjoint) union of all simple vertex groups, all residual vertex groups, and a subset of clean vertex groups. 
Recall that $\I_s$, $\I_r$ and $\I_c^+$ are the corresponding index sets of these groups.

By \Cref{clm:simple,clm:residual,clm:clean}, we have that
$|\I_s| + |\I_r| + |\I_c^+| \leq 5\cdot \sizeopt$.
Finally, since the size of each vertex group is at most $2 \alpha$, we can bound the size of the solution as follows: 

\[
\pushQED{\qed} 
|V_C| = \sum_{i \in \I_s \cup \I_r \cup \I_{c}^{+}} |V_i| \leq 2\alpha \cdot (|\I_s| + |\I_r| + |\I_{c}^{+}|) \leq 10\alpha \cdot \sizeopt. \qedhere
\]
\end{proof}


It remains to show that the algorithm can be implemented using $O(\nicefrac{n^2}{\alpha^2})$ bits of space. 
\Cref{alg:optvc} randomly partitions $V$, maintains several instances of \algms (\Cref{lem:ms}) and \algnt (\Cref{prop:Ntester}), and stores various counters. 
To show the space usage of the algorithm, we first consider each of these components separately.

\begin{claim} \label{clm:partVspace}
The partitioning of $V$ into $\frac{n}{\alpha}$ vertex groups of size in the range $[\alpha/2, 2\alpha]$ uses $O(\nicefrac{n^2}{\alpha^2})$ bits of space and succeeds with high probability.
\end{claim}
\begin{proof}
We show that for small $\alpha < \log^2 n$, i.e, when we have sufficient space, we can achieve this with a uniform random permutation, and for large $\alpha \geq \log^2 n$, we use a  ($10 \cdot \log n$)-wise independent hash function.

\subparagraph{Small $\alpha$.}
For any $\alpha < \log^2 n$, we can randomly permute the vertices using $O(n \log n) = O(\nicefrac{n^2}{\alpha^2})$ random bits to create a uniform random partitioning of $V$ into $\frac{n}{\alpha}$ groups of size $\alpha$.

\subparagraph{Large $\alpha$.}
For any $\alpha\geq \log^2 n$, we can partition $V$ using a ($10 \cdot \log n$)-wise independent hash function $h: \left[n\right] \rightarrow \left[\frac{n}{\alpha}\right]$ which uses $O(\log^2 n) = O(\nicefrac{n^2}{\alpha^2})$ bits by \Cref{prop:k-wise}. 
We bound the size of the groups as follows:
Consider any group $V_j$ ($j \in \left[\frac{n}{\alpha}\right]$) and let $X_i$ be the random variable that is $1$ if vertex $i$ is hashed to $V_j$, i.e., $h(i)=j$. 
Let $X=\sum_i X_i$ represent the number of vertices in group $V_j$. We have $\expect{X}=n \cdot (\alpha/n)=\alpha$. 
Using \Cref{prop:chernoff-limited} with $\eps=0.1$,
	\begin{align*}
		\prob{\card{X - \expect{X}} \geq \eps \cdot \expect{X}} 
        &\leq \exp\paren{-5\log n} 
        \leq n^{-5}.
	\end{align*}
A union bound over all groups implies that with probability at least $1-n^{-4}$, all groups have size between $0.9\alpha$ and $1.1\alpha$.
\end{proof}

\begin{claim} \label{clm:MNspace}
The instances of \algms and \algnt, and the counters use $O(\nicefrac{n^2}{\alpha^2})$ bits of space.
\end{claim}
\begin{proof}
We use one instance of \algms (\Cref{lem:ms}) which takes space $O(\nicefrac{n^2}{\alpha^2})$ bits.
We use $\nicefrac{n}{\alpha}$ instances of \algnt (\Cref{prop:Ntester}) with 
parameters $a=\nicefrac{n}{\alpha}$ and $b= n^{\nicefrac{\delta}{2}}$ each of which take space $O(\frac{a}{b} 
\log^3 n) = O( (\nicefrac{n}{\alpha}) \cdot (\log^3 n/ n^{\nicefrac{\delta}{2}})) = o(\nicefrac{n}{\alpha})$ bits. 
This implies that the total 
space used by $\nicefrac{n}{\alpha}$ instances is $O(\nicefrac{n^2}{\alpha^2})$ bits.
We maintain counters modulo a constant $c=\nicefrac{15}{\delta}$ for the number of edges between every pair of 
vertex groups. 
Each takes $O(1)$ bits of space, and since there are $O(\nicefrac{n^2}{\alpha^2})$ many 
of these counters, this totals $O(\nicefrac{n^2}{\alpha^2})$ bits of space. 
We also maintain counters for 
the number of internal edges for each group which requires $O(\log n)=o(\nicefrac{n}{\alpha})$ bits of space 
each. 
Since there are $\frac{n}{\alpha}$ many groups, this totals $O(\nicefrac{n^2}{\alpha^2})$ bits of space.
\end{proof}

Hence, by \Cref{clm:partVspace,clm:MNspace}, we have shown that the components of \Cref{alg:optvc} use $O(\nicefrac{n^2}{\alpha^2})$ bits in total. We still, however, need to consider the format of the output. 
When $\alpha$ gets large enough, the space is only $o(n)$, whereas simply storing the output -- the vertices of a solution -- could require $\Theta(n)$ bits of space. We solve this by showing that we can implicitly store the solution when there is limited space.
\begin{claim}\label{clm:outspace}
The output of \Cref{alg:optvc} can be maintained using $O(\nicefrac{n^2}{\alpha^2})$ bits of space.
\end{claim}
\begin{proof}
For $\alpha < \log^2 n$, we can maintain the vertices of the solution explicitly. 
For $\alpha \geq \log^2 n$, we rely on the hash function $h$ used to partition $V$ (see \Cref{clm:partVspace}). 
Recall that vertices are added to the solution at a group level, so we can simply maintain a bit vector of length $\frac{n}{\alpha}$ representing the groups added to the solution. 
Then, the output consists of $h$ and the bit vector which is sufficient for checking if a vertex belongs to the solution and uses $O(\log^2 n + \frac{n}{\alpha})=O(\nicefrac{n^2}{\alpha^2})$ bits of space.
\end{proof}


We have now shown that \Cref{alg:optvc} can be implemented using $O(\nicefrac{n^2}{\alpha^2})$ bits of space. Therefore, combined with \Cref{lem:vcoptapprox} and \Cref{asm:optapprox}, we have proven our main result, \Cref{thm:vcalg-full}.

\section{Conclusion}
\label{sec:concl}

In this paper, we have resolved the space complexity of \mvc for the full range of $\alpha$. 
We have provided a randomised algorithm which asymptotically matches the lower bound \cite{dk20} up to constant factors, showing that $\Theta(\nicefrac{n^2}{\alpha^2})$ is necessary and sufficient for this problem.

The previous best algorithm for \mvc was a deterministic one using $O(\nicefrac{n^2}{\alpha^2} \cdot \log \alpha)$ bits of space \cite{dk20}. 
We have shown that we can remove the logarithmic overhead using randomness.
Can we, however, remove this logarithmic factor using deterministic techniques or otherwise prove a deterministic lower bound which shows that it is necessary?

Our work continues the direction set by the results on connectivity \cite{agm12,nelson2019optimal} and matchings \cite{dk20,as22}; we resolve the space complexity (up to constant factors) of another problem in the dynamic graph streaming setting.
However, other problems still remain open.
Hence, can we achieve this for other dynamic graph streaming problems such as dominating set \cite{kk22} and spectral sparsification \cite{kmmmnst20}?

\section*{Acknowledgements}

We are grateful to Sepehr Assadi and Christian Konrad for many helpful discussions.
We also appreciate the valuable comments from our APPROX 2022 reviewers.

\pagebreak
\bibliography{main}

\newcommand{\etalchar}[1]{$^{#1}$}
\begin{thebibliography}{KMM{\etalchar{+}}20}

\bibitem[ABB{\etalchar{+}}19]{abbms19}
Sepehr Assadi, MohammadHossein Bateni, Aaron Bernstein, Vahab~S. Mirrokni, and
  Cliff Stein.
\newblock Coresets meet {EDCS:} algorithms for matching and vertex cover on
  massive graphs.
\newblock In Timothy~M. Chan, editor, {\em Proceedings of the Thirtieth Annual
  {ACM-SIAM} Symposium on Discrete Algorithms, {SODA} 2019, San Diego,
  California, USA, January 6-9, 2019}, pages 1616--1635. {SIAM}, 2019.

\bibitem[ACG{\etalchar{+}}15]{ahn2015correlation}
KookJin Ahn, Graham Cormode, Sudipto Guha, Andrew McGregor, and Anthony Wirth.
\newblock Correlation clustering in data streams.
\newblock In {\em International Conference on Machine Learning}, pages
  2237--2246. PMLR, 2015.

\bibitem[ACK19a]{ack19a}
Sepehr Assadi, Yu~Chen, and Sanjeev Khanna.
\newblock Polynomial pass lower bounds for graph streaming algorithms.
\newblock In Moses Charikar and Edith Cohen, editors, {\em Proceedings of the
  51st Annual {ACM} {SIGACT} Symposium on Theory of Computing, {STOC} 2019,
  Phoenix, AZ, USA, June 23-26, 2019}, pages 265--276. {ACM}, 2019.

\bibitem[ACK19b]{assadi2019sublinear}
Sepehr Assadi, Yu~Chen, and Sanjeev Khanna.
\newblock Sublinear algorithms for ($\delta$+ 1) vertex coloring.
\newblock In {\em Proceedings of the Thirtieth Annual ACM-SIAM Symposium on
  Discrete Algorithms}, pages 767--786. SIAM, 2019.

\bibitem[AGM12a]{agm12}
Kook~Jin Ahn, Sudipto Guha, and Andrew McGregor.
\newblock Analyzing graph structure via linear measurements.
\newblock In Yuval Rabani, editor, {\em Proceedings of the Twenty-Third Annual
  {ACM-SIAM} Symposium on Discrete Algorithms, {SODA} 2012, Kyoto, Japan,
  January 17-19, 2012}, pages 459--467. {SIAM}, 2012.

\bibitem[AGM12b]{agm12b}
Kook~Jin Ahn, Sudipto Guha, and Andrew McGregor.
\newblock Graph sketches: sparsification, spanners, and subgraphs.
\newblock In Michael Benedikt, Markus Kr{\"{o}}tzsch, and Maurizio Lenzerini,
  editors, {\em Proceedings of the 31st {ACM} {SIGMOD-SIGACT-SIGART} Symposium
  on Principles of Database Systems, {PODS} 2012, Scottsdale, AZ, USA, May
  20-24, 2012}, pages 5--14. {ACM}, 2012.

\bibitem[AJJ{\etalchar{+}}22]{ajjst22}
Sepehr Assadi, Arun Jambulapati, Yujia Jin, Aaron Sidford, and Kevin Tian.
\newblock Semi-streaming bipartite matching in fewer passes and optimal space.
\newblock In Joseph~(Seffi) Naor and Niv Buchbinder, editors, {\em Proceedings
  of the 2022 {ACM-SIAM} Symposium on Discrete Algorithms, {SODA} 2022, Virtual
  Conference / Alexandria, VA, USA, January 9 - 12, 2022}, pages 627--669.
  {SIAM}, 2022.

\bibitem[AKLY16]{akly16}
Sepehr Assadi, Sanjeev Khanna, Yang Li, and Grigory Yaroslavtsev.
\newblock Maximum matchings in dynamic graph streams and the simultaneous
  communication model.
\newblock In Robert Krauthgamer, editor, {\em Proceedings of the Twenty-Seventh
  Annual {ACM-SIAM} Symposium on Discrete Algorithms, {SODA} 2016, Arlington,
  VA, USA, January 10-12, 2016}, pages 1345--1364. {SIAM}, 2016.

\bibitem[AKSY20]{aksy20}
Sepehr Assadi, Gillat Kol, Raghuvansh~R. Saxena, and Huacheng Yu.
\newblock Multi-pass graph streaming lower bounds for cycle counting, max-cut,
  matching size, and other problems.
\newblock In Sandy Irani, editor, {\em 61st {IEEE} Annual Symposium on
  Foundations of Computer Science, {FOCS} 2020, Durham, NC, USA, November
  16-19, 2020}, pages 354--364. {IEEE}, 2020.

\bibitem[AS22]{as22}
Sepehr Assadi and Vihan Shah.
\newblock An asymptotically optimal algorithm for maximum matching in dynamic
  streams.
\newblock In Mark Braverman, editor, {\em 13th Innovations in Theoretical
  Computer Science Conference, {ITCS} 2022, January 31 - February 3, 2022,
  Berkeley, CA, {USA}}, volume 215 of {\em LIPIcs}, pages 9:1--9:23. Schloss
  Dagstuhl - Leibniz-Zentrum f{\"{u}}r Informatik, 2022.

\bibitem[Ass22]{a22}
Sepehr Assadi.
\newblock A two-pass (conditional) lower bound for semi-streaming maximum
  matching.
\newblock In Joseph~(Seffi) Naor and Niv Buchbinder, editors, {\em Proceedings
  of the 2022 {ACM-SIAM} Symposium on Discrete Algorithms, {SODA} 2022, Virtual
  Conference / Alexandria, VA, USA, January 9 - 12, 2022}, pages 708--742.
  {SIAM}, 2022.

\bibitem[BdBM21]{bbm21}
Leyla Biabani, Mark de~Berg, and Morteza Monemizadeh.
\newblock Maximum-weight matching in sliding windows and beyond.
\newblock In Hee{-}Kap Ahn and Kunihiko Sadakane, editors, {\em 32nd
  International Symposium on Algorithms and Computation, {ISAAC} 2021, December
  6-8, 2021, Fukuoka, Japan}, volume 212 of {\em LIPIcs}, pages 73:1--73:16.
  Schloss Dagstuhl - Leibniz-Zentrum f{\"{u}}r Informatik, 2021.

\bibitem[Ber20]{b20}
Aaron Bernstein.
\newblock Improved bounds for matching in random-order streams.
\newblock In Artur Czumaj, Anuj Dawar, and Emanuela Merelli, editors, {\em 47th
  International Colloquium on Automata, Languages, and Programming, {ICALP}
  2020, July 8-11, 2020, Saarbr{\"{u}}cken, Germany (Virtual Conference)},
  volume 168 of {\em LIPIcs}, pages 12:1--12:13. Schloss Dagstuhl -
  Leibniz-Zentrum f{\"{u}}r Informatik, 2020.

\bibitem[BKSV14]{braverman2014optimal}
Vladimir Braverman, Jonathan Katzman, Charles Seidell, and Gregory Vorsanger.
\newblock An optimal algorithm for large frequency moments using o
  (n\^{}(1-2/k)) bits.
\newblock In {\em Approximation, Randomization, and Combinatorial Optimization.
  Algorithms and Techniques (APPROX/RANDOM 2014)}. Schloss
  Dagstuhl-Leibniz-Zentrum fuer Informatik, 2014.

\bibitem[CCE{\etalchar{+}}16]{ccehmmv16}
Rajesh Chitnis, Graham Cormode, Hossein Esfandiari, MohammadTaghi Hajiaghayi,
  Andrew McGregor, Morteza Monemizadeh, and Sofya Vorotnikova.
\newblock Kernelization via sampling with applications to finding matchings and
  related problems in dynamic graph streams.
\newblock In Robert Krauthgamer, editor, {\em Proceedings of the Twenty-Seventh
  Annual {ACM-SIAM} Symposium on Discrete Algorithms, {SODA} 2016, Arlington,
  VA, USA, January 10-12, 2016}, pages 1326--1344. {SIAM}, 2016.

\bibitem[CMS13]{cms13}
Michael~S. Crouch, Andrew McGregor, and Daniel~M. Stubbs.
\newblock Dynamic graphs in the sliding-window model.
\newblock In Hans~L. Bodlaender and Giuseppe~F. Italiano, editors, {\em
  Algorithms - {ESA} 2013 - 21st Annual European Symposium, Sophia Antipolis,
  France, September 2-4, 2013. Proceedings}, volume 8125 of {\em Lecture Notes
  in Computer Science}, pages 337--348. Springer, 2013.

\bibitem[CS14]{cs14}
Michael~S. Crouch and Daniel~M. Stubbs.
\newblock Improved streaming algorithms for weighted matching, via unweighted
  matching.
\newblock In Klaus Jansen, Jos{\'{e}} D.~P. Rolim, Nikhil~R. Devanur, and
  Cristopher Moore, editors, {\em Approximation, Randomization, and
  Combinatorial Optimization. Algorithms and Techniques, {APPROX/RANDOM} 2014,
  September 4-6, 2014, Barcelona, Spain}, volume~28 of {\em LIPIcs}, pages
  96--104. Schloss Dagstuhl - Leibniz-Zentrum f{\"{u}}r Informatik, 2014.

\bibitem[DK20]{dk20}
Jacques Dark and Christian Konrad.
\newblock Optimal lower bounds for matching and vertex cover in dynamic graph
  streams.
\newblock In Shubhangi Saraf, editor, {\em 35th Computational Complexity
  Conference, {CCC} 2020, July 28-31, 2020, Saarbr{\"{u}}cken, Germany (Virtual
  Conference)}, volume 169 of {\em LIPIcs}, pages 30:1--30:14. Schloss Dagstuhl
  - Leibniz-Zentrum f{\"{u}}r Informatik, 2020.

\bibitem[DP09]{dubhashi2009concentration}
Devdatt~P Dubhashi and Alessandro Panconesi.
\newblock {\em Concentration of measure for the analysis of randomized
  algorithms}.
\newblock Cambridge University Press, 2009.

\bibitem[FKM{\etalchar{+}}04]{fkmsz04}
Joan Feigenbaum, Sampath Kannan, Andrew McGregor, Siddharth Suri, and Jian
  Zhang.
\newblock On graph problems in a semi-streaming model.
\newblock In Josep Diaz, Juhani Karhum{\"{a}}ki, Arto Lepist{\"{o}}, and Donald
  Sannella, editors, {\em Automata, Languages and Programming: 31st
  International Colloquium, {ICALP} 2004, Turku, Finland, July 12-16, 2004.
  Proceedings}, volume 3142 of {\em Lecture Notes in Computer Science}, pages
  531--543. Springer, 2004.

\bibitem[GKMS19]{gkms18ext}
Buddhima Gamlath, Sagar Kale, Slobodan Mitrovic, and Ola Svensson.
\newblock Weighted matchings via unweighted augmentations.
\newblock In Peter Robinson and Faith Ellen, editors, {\em Proceedings of the
  2019 {ACM} Symposium on Principles of Distributed Computing, {PODC} 2019,
  Toronto, ON, Canada, July 29 - August 2, 2019}, pages 491--500. {ACM}, 2019.

\bibitem[IPW11]{indyk2011power}
Piotr Indyk, Eric Price, and David~P Woodruff.
\newblock On the power of adaptivity in sparse recovery.
\newblock In {\em 2011 IEEE 52nd Annual Symposium on Foundations of Computer
  Science}, pages 285--294. IEEE, 2011.

\bibitem[JST11]{jst11}
Hossein Jowhari, Mert Saglam, and G{\'{a}}bor Tardos.
\newblock Tight bounds for lp samplers, finding duplicates in streams, and
  related problems.
\newblock In Maurizio Lenzerini and Thomas Schwentick, editors, {\em
  Proceedings of the 30th {ACM} {SIGMOD-SIGACT-SIGART} Symposium on Principles
  of Database Systems, {PODS} 2011, June 12-16, 2011, Athens, Greece}, pages
  49--58. {ACM}, 2011.

\bibitem[KK22]{kk22}
Sanjeev Khanna and Christian Konrad.
\newblock Optimal bounds for dominating set in graph streams.
\newblock In Mark Braverman, editor, {\em 13th Innovations in Theoretical
  Computer Science Conference, {ITCS} 2022, January 31 - February 3, 2022,
  Berkeley, CA, {USA}}, volume 215 of {\em LIPIcs}, pages 93:1--93:23. Schloss
  Dagstuhl - Leibniz-Zentrum f{\"{u}}r Informatik, 2022.

\bibitem[KLM{\etalchar{+}}14]{klmms14}
Michael Kapralov, Yin~Tat Lee, Cameron Musco, Christopher Musco, and Aaron
  Sidford.
\newblock Single pass spectral sparsification in dynamic streams.
\newblock In {\em 55th {IEEE} Annual Symposium on Foundations of Computer
  Science, {FOCS} 2014, Philadelphia, PA, USA, October 18-21, 2014}, pages
  561--570. {IEEE} Computer Society, 2014.

\bibitem[KMM12]{kmm12}
Christian Konrad, Fr{\'{e}}d{\'{e}}ric Magniez, and Claire Mathieu.
\newblock Maximum matching in semi-streaming with few passes.
\newblock In Anupam Gupta, Klaus Jansen, Jos{\'{e}} D.~P. Rolim, and Rocco~A.
  Servedio, editors, {\em Approximation, Randomization, and Combinatorial
  Optimization. Algorithms and Techniques - 15th International Workshop,
  {APPROX} 2012, and 16th International Workshop, {RANDOM} 2012, Cambridge, MA,
  USA, August 15-17, 2012. Proceedings}, volume 7408 of {\em Lecture Notes in
  Computer Science}, pages 231--242. Springer, 2012.

\bibitem[KMM{\etalchar{+}}20]{kmmmnst20}
Michael Kapralov, Aida Mousavifar, Cameron Musco, Christopher Musco, Navid
  Nouri, Aaron Sidford, and Jakab Tardos.
\newblock Fast and space efficient spectral sparsification in dynamic streams.
\newblock In Shuchi Chawla, editor, {\em Proceedings of the 2020 {ACM-SIAM}
  Symposium on Discrete Algorithms, {SODA} 2020, Salt Lake City, UT, USA,
  January 5-8, 2020}, pages 1814--1833. {SIAM}, 2020.

\bibitem[KN21]{kn21}
Christian Konrad and Kheeran~K. Naidu.
\newblock On two-pass streaming algorithms for maximum bipartite matching.
\newblock In Mary Wootters and Laura Sanit{\`{a}}, editors, {\em Approximation,
  Randomization, and Combinatorial Optimization. Algorithms and Techniques,
  {APPROX/RANDOM} 2021, August 16-18, 2021, University of Washington, Seattle,
  Washington, {USA} (Virtual Conference)}, volume 207 of {\em LIPIcs}, pages
  19:1--19:18. Schloss Dagstuhl - Leibniz-Zentrum f{\"{u}}r Informatik, 2021.

\bibitem[KNP{\etalchar{+}}17]{KapralovNPWWY17}
Michael Kapralov, Jelani Nelson, Jakub Pachocki, Zhengyu Wang, David~P.
  Woodruff, and Mobin Yahyazadeh.
\newblock Optimal lower bounds for universal relation, and for samplers and
  finding duplicates in streams.
\newblock In Chris Umans, editor, {\em 58th {IEEE} Annual Symposium on
  Foundations of Computer Science, {FOCS} 2017, Berkeley, CA, USA, October
  15-17, 2017}, pages 475--486. {IEEE} Computer Society, 2017.

\bibitem[Kon15]{k15}
Christian Konrad.
\newblock Maximum matching in turnstile streams.
\newblock In Nikhil Bansal and Irene Finocchi, editors, {\em Algorithms - {ESA}
  2015 - 23rd Annual European Symposium, Patras, Greece, September 14-16, 2015,
  Proceedings}, volume 9294 of {\em Lecture Notes in Computer Science}, pages
  840--852. Springer, 2015.

\bibitem[Kon18a]{konrad2018mis}
Christian Konrad.
\newblock Mis in the congested clique model in $ o(\log \log \delta) $ rounds.
\newblock {\em arXiv preprint arXiv:1802.07647}, 2018.

\bibitem[Kon18b]{k18}
Christian Konrad.
\newblock A simple augmentation method for matchings with applications to
  streaming algorithms.
\newblock In Igor Potapov, Paul~G. Spirakis, and James Worrell, editors, {\em
  43rd International Symposium on Mathematical Foundations of Computer Science,
  {MFCS} 2018, August 27-31, 2018, Liverpool, {UK}}, volume 117 of {\em
  LIPIcs}, pages 74:1--74:16. Schloss Dagstuhl - Leibniz-Zentrum f{\"{u}}r
  Informatik, 2018.

\bibitem[Kon21]{konrad2021frequent}
Christian Konrad.
\newblock Frequent elements with witnesses in data streams.
\newblock In {\em Proceedings of the 40th ACM SIGMOD-SIGACT-SIGAI Symposium on
  Principles of Database Systems}, pages 83--95, 2021.

\bibitem[McG14]{m14}
Andrew McGregor.
\newblock Graph stream algorithms: a survey.
\newblock {\em {SIGMOD} Rec.}, 43(1):9--20, 2014.

\bibitem[MR95]{MotwaniR95}
Rajeev Motwani and Prabhakar Raghavan.
\newblock {\em Randomized Algorithms}.
\newblock Cambridge University Press, 1995.

\bibitem[Mut05]{m05}
S.~Muthukrishnan.
\newblock Data streams: Algorithms and applications.
\newblock {\em Found. Trends Theor. Comput. Sci.}, 1(2), 2005.

\bibitem[NY19]{nelson2019optimal}
Jelani Nelson and Huacheng Yu.
\newblock Optimal lower bounds for distributed and streaming spanning forest
  computation.
\newblock In {\em Proceedings of the Thirtieth Annual ACM-SIAM Symposium on
  Discrete Algorithms}, pages 1844--1860. SIAM, 2019.

\bibitem[PW13]{price2013lower}
Eric Price and David~P Woodruff.
\newblock Lower bounds for adaptive sparse recovery.
\newblock In {\em Proceedings of the twenty-fourth annual ACM-SIAM symposium on
  Discrete algorithms}, pages 652--663. SIAM, 2013.

\bibitem[SSS95]{SchmidtSS95}
Jeanette~P. Schmidt, Alan Siegel, and Aravind Srinivasan.
\newblock Chernoff-hoeffding bounds for applications with limited independence.
\newblock {\em {SIAM} J. Discret. Math.}, 8(2):223--250, 1995.

\bibitem[SW15]{sun2015tight}
Xiaoming Sun and David~P Woodruff.
\newblock Tight bounds for graph problems in insertion streams.
\newblock In {\em Approximation, Randomization, and Combinatorial Optimization.
  Algorithms and Techniques (APPROX/RANDOM 2015)}. Schloss
  Dagstuhl-Leibniz-Zentrum fuer Informatik, 2015.

\bibitem[WZ12]{woodruff2012tight}
David~P Woodruff and Qin Zhang.
\newblock Tight bounds for distributed functional monitoring.
\newblock In {\em Proceedings of the forty-fourth annual ACM symposium on
  Theory of computing}, pages 941--960, 2012.

\end{thebibliography}
\pagebreak

\appendix

\section{An altered Match-or-Sparsify}
\label{app:m-or-s}



In this section, we prove that an altered version of Assadi and Shah's \morsn implies \Cref{lem:ms}. Informally speaking, given a graph $G$, this lemma gives an algorithm that either finds a large matching in $G$ or identifies a sparse induced subgraph of $G$. 

The analysis we give follows the analysis of Assadi and Shah \cite{as22}, and we give the full argument for completeness. We will also highlight the key alterations that we make. Before beginning, we give some required probabilistic and sketching tools in \Cref{app:probtools} and \Cref{app:sketchtools}, respectively.

\subsection{Probabilistic Tools} \label{app:probtools}

 \begin{proposition}[Chernoff bound; c.f.~\cite{dubhashi2009concentration}]\label{prop:chernoff}
 	Suppose $X_1,\ldots,X_m$ are $m$ independent random variables with range $[0,1]$ each. 
	 Let $X := \sum_{i=1}^m X_i$ and $\mu_L \leq \expect{X} \leq \mu_H$. 
	 Then, for any $\eps > 0$, 
 	\[
 	\prob{X >  (1+\eps) \cdot \mu_H} \leq \exp\paren{-\frac{\eps^2 \cdot \mu_H}{3+\eps}} \quad \textnormal{and} \quad \prob{X <  (1-\eps) \cdot \mu_L} \leq \exp\paren{-\frac{\eps^2 \cdot \mu_L}{2+\eps}}.
 	\]
 \end{proposition}
 
 
\subsection{Sketching Tools}\label{app:sketchtools} 

We present the following sketching tool by Assadi and Shah \cite{as22} for the neighbourhood edge sampling problem.
\begin{problem}\label{prob:nei-edge-sampler}
	Given a graph $G=(V,E)$ specified in a dynamic stream, and a set $S \subseteq V$ of vertices at the start of the stream, output an edge $(u,v)$ such that $u \in S$ and $v$ is sampled uniformly at random from $N(S)$. 
\end{problem}
\noindent We will use the following linear sketch for solving this problem.
\begin{proposition}[\cite{as22}]
\label{lem:nei-edge-sampler}
	There is a linear sketch, called $\NES(G,S)$, for~\Cref{prob:nei-edge-sampler} with size 
	\[
	\snes= \snes(n) = O(\log^3{n})
	\]
	 bits, that outputs \textnormal{FAIL} with probability at most ${1}/{100}$ and gives a wrong answer with probability at most $n^{-8}$.
\end{proposition}

\subsection{Proof of \texorpdfstring{\Cref{lem:ms}}{Lemma}}


\ms*

The algorithm in~\Cref{lem:ms} samples $\approx n^2/\alpha^2 \cdot \log^3{n}$ edges from the graph using a non-uniform distribution as follows: 
for each sample, first pick $\approx n/\alpha$ vertices $S$ uniformly at random and then use $\NES$ to sample an edge from $S$ to a vertex of $N(S)$ chosen uniformly at random. 
Given the bound of $O(\log^3{n})$ bits on the size of sketches for $\NES$, the total space of the algorithm can be bounded by $O(n^2/\alpha^2)$ bits.
In the recovery phase then, a greedy matching is computed over these sampled edges and is returned as $\Measy$. 
Note that in the analysis it will be helpful to think of the edges being recovered one by one and fed to the greedy matching algorithm. The key change we make here is in the number of edge samples taken and that we fix the parameter $\parms = n$ (see \Cref{sec:randomgreedy} for a discussion of $\parms$). Formally, we have \Cref{alg:ms}.


	
	
	
	


	
	
	
	

	

\begin{algorithm*}[ht]
\caption{Altered Match-Or-Sparsify Lemma~(\Cref{lem:ms})} 
\label{alg:ms} 
\textbf{Input:} A dynamic graph stream $\sigma$ for a $n$-vertex graph $G =(V,E)$, a small 
constant $\delta > 0$, and a positive integer $\alpha \leq n^{1-\delta}$ 

\textbf{Output:} A matching $\Measy$ in $G$
\vspace{2mm} 
 
\textbf{Pre-processing:}
\begin{algorithmic}[1]
\State Let $k := {n}/{\alpha}$ and $s := n^2/\alpha^2 \cdot \log^3{n}$
\For{$i \in [2s]$ \footnotemark}
    \State Sample a pair-wise independent hash function $h_i : V \rightarrow [k]$ \label{ms:a}
    \State Set $V_i~:=~\set{v \in V \mid h_i(v) = 1}$
    \State Initialise $\mathbf{N}_i$ to be an instance of $\NES(G,V_i)$ \label{ms:b}
\EndFor
\end{algorithmic}
\vspace{2mm}
\textbf{Processing the stream:}
\begin{algorithmic}[1]
\setcounter{ALG@line}{\getrefnumber{optvc:countend1}}
\State Update each $\mathbf{N}_i$ using $\sigma$

\end{algorithmic}
\vspace{2mm}
\textbf{Post-processing:}
\begin{algorithmic}[1]
\setcounter{ALG@line}{\getrefnumber{optvc:countend2}}
\State For all $i \in [2s]$, run the recovery algorithm of $\NES(G,V_i)$ on $\mathbf{N}_i$ to get an output edge $e_i$ (we write $e_i = \perp$ if the sampler outputs FAIL)
\State Let $\Measy$ be the greedy matching over the sampled edges $e_i$.
\State \textbf{return} $\Measy$
\end{algorithmic}
\end{algorithm*}

\footnotetext{The steps are partitioned into two \textbf{batches} of $s$ steps each in the analysis, hence the use of $2s$ for the number of steps.}

We first bound the space of the algorithm.

\begin{lemma}\label{lem:ms-space}
	\Cref{alg:ms} uses $O(n^2/\alpha^2)$ bits of space with high probability. 
\end{lemma}
\begin{proof}
 
In each step, \Cref{ms:a} requires storing a pair-wise independent hash function which needs $O(\log n)$ bits of space by~\Cref{prop:k-wise}. \Cref{ms:b} requires storing an $\NES$ which needs $O(\log^3 n)$ bits by \Cref{lem:nei-edge-sampler}.
There are $2s = O(n^2/\alpha^2 \cdot \log^3{n})$ steps, so the total space  $O(n^2/\alpha^2)$ bits.
\end{proof}

We now prove that the matching $\Measy$ output by \Cref{alg:ms} satisfies the guarantees of \Cref{lem:ms}. 
For simplicity, we use similar notations and definitions as the work by Assadi and Shah \cite{as22}.

\paragraph{Notation.} For any $i \in [2s]$, let $\ME{i}$  be the set of edges included in $\Measy$ in the first $i-1$ steps of the recovery, i.e., from $\set{e_j}_{j=1}^{i-1}$, 
and $\GE{i}$ to be the subgraph of $G$ induced on \emph{unmatched} vertices of $\ME{i}$.
We use $\degG{v}{i}$ to denote the degree of each vertex in $\GE{i}$ to other vertices in $\GE{i}$. 
We partition vertices of $\GE{i}$ based on their degrees in $\GE{i}$ into \textbf{low-}, \textbf{medium-}, and \textbf{high-degree} as follows: 
\begin{align*}
 \Low{i} &:= \set{v: \degG{v}{i} < (\log^3\!{n})}, \; \Med{i} := \set{v : (\log^3\!{n}) \leq \degG{v}{i} < (\frac{n}{8\alpha})}, \\
  \High{i} &:= \set{v :  \degG{v}{i} \geq \frac{n}{8\alpha}}.
\end{align*}

We define the following two events: 
\begin{itemize}
	\item $\eventmatch(i)$: the matching $\ME{i}$  has less than $(n/8\alpha)$ edges (i.e., matching-case not happened); 
	\item $\eventsparsify(i)$: the subgraph $\GE{i}$ has more than $(20\,n \cdot \log^4\!{n}/\alpha)$ edges (i.e., sparsify-case not happened). 
\end{itemize}
Finally, we say that a choice of $V_i$ in step $i \in [2s]$ is \textbf{clean} if $V_i$ does not contain any matched vertices of $\ME{i}$. 

The key change we make here is that, in our partitioning of $\GE{i}$, the boundary between \textbf{low-} and \textbf{medium-} is reduced by an $\alpha$-factor. Coupled with the increase in number of samples, this allows us to also change the definition of event $\eventsparsify(i)$ to reflect the $\alpha$-factor increase in sparseness guarantees which we require.

The events are defined in such a way that if at least one of these events do not happen for some $i \in [2s]$, then~\Cref{alg:ms} succeeds in outputting the desired matching of~\Cref{lem:ms}. The formal argument is as follows.
\begin{claim}\label{clm:event-ms}
	Suppose for some $i \in [2s]$, either of $\eventmatch(i)$ or $\eventsparsify(i)$ does not happen; then,  $\Measy$ of~\Cref{alg:ms} satisfies the guarantees of~\Cref{lem:ms}. 
\end{claim}
\begin{proof}
	Suppose first that $\eventmatch(i)$ does not happen. This means $\ME{i}$ has size at least $(n/8\alpha)$ and by the greedy choice of $\Measy$, we have $\card{\Measy} \geq \card{\ME{i}} \geq (n/8\alpha)$,  satisfying the match-case condition. 
		
	Now suppose that $\eventsparsify(i)$ does not happen. Since the number of edges of $\Geasy$ can only be smaller than that of $\GE{i}$, we have that $\Geasy$ also only has $(20\,n\cdot\log^4\!{n}/\alpha)$ edges, satisfying the sparsify-case condition. 
\end{proof}

The idea is to show that with high probability, for some $i \in [2s]$, one of the events $\eventmatch(i)$ or $\eventsparsify(i)$ is \emph{not} going to happen, then \Cref{lem:ms} holds by \Cref{lem:ms-space} and \Cref{clm:event-ms}. 
In order to do this, we need to show that, for any step $i$, there will be a probability of $\approx k/s$ in increasing the size of $\ME{i}$ by one as long as both $\eventmatch(i)$ and $\eventsparsify(i)$ happen. Therefore, if these events both happen for at least $s$ steps, it will ultimately lead to event $\eventmatch(2s)$ not happening, i.e., a large matching $\Measy$.

The variance reduction ideas used to prove this rely on the set $\High{i}$ being empty for each step $i$. Hence, to achieve this, the steps of the algorithm are partitioned into two \textbf{batches} each of size $s$. The analysis of the first batch shows that $\High{i}$ is empty for every step in the second batch, i.e., for all $i \in (s, 2s]$. Then, the second batch is used for the main argument. We will prove the following:
\begin{itemize}[leftmargin=10pt]
	\item \textbf{First batch:} As long as $\eventmatch(i)$ and $\eventsparsify(i)$ happen for all $i \in [s]$, with high probability, the set $\High{s+1}$ (and thus $\High{j}$ for all $j \in (s,2s]$) will be empty for the second batch. 

\begin{lemma}\label{lem:ms-1st-batch}
	With high probability, either at least one of $\eventmatch(i)$ and $\eventsparsify(i)$ does not happen for some step $i \in [s]$ or $\High{s+1}$ will be empty.
\end{lemma}

\item \textbf{Second batch:} Whenever both $\eventmatch(i)$ and $\eventsparsify(i)$ happen in a step $i \in (s , 2s]$, there will be a probability of $\approx k/s$ in increasing 
the size of $\ME{i}$ by one in this step. Given that this process is repeated for $s$ steps, $\Measy$ will eventually become of size $\approx k = n/\alpha$ (or one of the events happen along the way, and \Cref{clm:event-ms} is used instead).

\begin{lemma}\label{lem:ms-2nd-batch}
	Assuming $\High{s+1}$ is empty, with high probability, at least one of the events $\eventmatch(i)$ or $\eventsparsify(i)$ does not happen for some $i \in (s:2s]$.
\end{lemma}

\end{itemize}

We now make the following remark.

\begin{remark}
{The actions of~\Cref{alg:ms} are clearly \emph{not} independent across different steps (in the recovery phase). However, in the upcoming probability analysis in each step $i \in [2s]$ {the randomness of all prior steps} conditioned on the events $\eventmatch(i)$ and $\eventsparsify(i)$ are fixed, and \emph{only} the randomness of the choice of $(V_i,e_i)$ are used in this step. This randomness is independent of prior steps. 
As such, in the following, \textbf{\underline{all} probability calculations in a step $i$ are conditioned on randomness of prior steps and events $\eventmatch(i)$ and $\eventsparsify(i)$}, without writing it explicitly each time. These 
probability calculations may not necessarily remain correct when either of these events do not happen, but we will be done by~\Cref{clm:event-ms} in those cases anyway.} 
\end{remark}

The following simple helper claim will be useful in the subsequent proofs (this claim would have been trivial had $h_i$ been a truly independent hash function instead of a pairwise-independent one).

\begin{claim}\label{clm:ms-helper-1}
	Consider any step $i \in [s]$ and let $v$ be any arbitrary vertex in $\GE{i}$. Then, 
	\[
		\Pr_{V_i}\Paren{v \in V_i ~\textnormal{and $V_i$ is clean}} \geq \frac{3}{4k}. 
	\]
\end{claim}
\begin{proof}
	Recall that there are at most $n/4\alpha = k/4$ vertices matched by $\ME{i}$. 	We have, 
	\begin{align*}
		\Pr_{V_i}\Paren{v \in V_i ~\textnormal{and $V_i$ is clean}} &= \Pr\paren{h_i(v) = 1} \cdot \Pr\paren{\text{$V_i$ is clean} \mid h_i(v)=1} \tag{$v \in V_i$ iff $h_i(v)=1$} \\
		&=  \frac1k \cdot \Paren{1-\Pr\paren{\text{$V_i$ is not clean}  \mid h_i(v)=1}} \tag{as $h_i(v) = 1$ w.p. $1/k$} \\
		&\geq \frac1k \cdot \paren{1-\sum_{u \in V(\ME{i})} \Pr\paren{h_i(u)=1 \mid h_i(v)=1}} \tag{by union bound and since $V_i$ is not clean iff $h_i(u)=1$ for some $u \in V(\ME{i})$} \\
		&= \frac1k \cdot \paren{1-\sum_{u \in V(\ME{i})} \frac1k} \tag{$h_i(\cdot)$ is a pairwise-independent hash function} \\
		&= \frac1k \cdot \paren{1-\frac{k}{4} \cdot \frac{1}{k}} \tag{as $h_i(u)=1$ w.p. $1/k$ and there are at most $k/4$ choices for matched vertices}, 
	\end{align*}
	which is at least $3/4k$ as desired. 
\end{proof}

In \Cref{subsec:1st-batch}, we follow exactly the steps taken by Assadi and Shah \cite{as22} and get the exact same intermediary results. The difference is crucially in the final arguments which rely on the number of samples $s$ taken. Their analysis rely on the assumption that $n \geq \alpha^2 \cdot n^\delta$ which does not hold in our case. This, however, is circumvented by the additional $\alpha$-factor in the number of samples taken.

In \Cref{subsec:2nd-batch}, we also follow their analysis closely; however, the change in definition of $\eventsparsify(i)$ for each step $i$ slightly alters the intermediary results. In essence, the probability that we increase the size of the matching in a particular step ($\approx k/s$) can be an $\alpha$-factor smaller since we have an $\alpha$-factor many more steps.
 
\subsection{First Batch: Proof of \texorpdfstring{\Cref{lem:ms-1st-batch}}{Lemma} } \label{subsec:1st-batch}

	Let $v$ be any vertex in $V$ and consider any step $i \in [s]$. If $\degG{v}{i} < (n/8\alpha)$, then $v$ cannot be part of $\High{i}$ and subsequently $\High{s+1}$ since $\GE{s+1}$ is a subgraph of $\GE{i}$. 
	In the following, we consider the case where $\degG{v}{i} \geq (n/8\alpha)$ and prove that there is a non-trivial chance of ``progress'' (to be defined later) in each step. 
	We first bound the probability of the following useful event for our analysis. 
	
	\begin{claim}\label{clm:ms-1st-v}
		In step $i$, if $\degG{v}{i} \geq n/8 \alpha$, we have $\Pr_{V_i}\paren{v \in N(V_i) ~\textnormal{and $V_i$ is clean} } \geq \dfrac{1}{16}$.
	\end{claim}
	\begin{proof}
		Let $d(v) := n/8\alpha$ and $D(v)$ be a set of $d$ arbitrary neighbors of $v$ in $\GE{i}$. We know that $v$ will be included in $N(V_i)$ if any of vertices in $D(v)$ is sampled in $V_i$. We have, 
		\begin{align*}
			\Pr_{V_i}\paren{v \in N(V_i) ~\textnormal{and $V_i$ is clean} } &\geq \Pr\paren{D(v) \cap V_i \neq \emptyset ~\text{and $V_i$ is clean}} \tag{$D(v) \subseteq N(v)$} \\
			&\geq \sum_{u \in D(v)} \Pr\paren{u \in V_i ~\text{and $V_i$ is clean}}  - \sum_{u\neq w \in D(v)} \Pr\paren{u,w \in V_i} \tag{by inclusion-exclusion principle and bounding 
			$\Pr\paren{u,w \in V_i} \geq \Pr\paren{u,w \in V_i~\text{and $V_i$ is clean}}$} \\
			&> \frac{3d(v)}{4k}  - \frac{d(v)^2}{k^2} \tag{by~\Cref{clm:ms-helper-1} and as $h_i(\cdot)$ is a pair-wise independent hash function with range $[k]$} \\
			&= \frac{(n/8\alpha)}{(n/\alpha)} \cdot \paren{\frac{3}{4} - \frac{(n/8\alpha)}{(n/\alpha)}}, \tag{as $d(v)=n/8\alpha$ and $k=n/\alpha$} 
		\end{align*}
		which is at least $1/16$ as desired. \Qed{\Cref{clm:ms-1st-v}}
		
	\end{proof}

	Let us now condition on the choice of $V_i$ and assume the event of~\Cref{clm:ms-1st-v} has happened. We say that this step $i$ is a \textbf{matching-step} if $N(V_i) > (n/2\alpha)$; otherwise, 
	we call this step a \textbf{vertex-step}. We argue that in a matching-step we have a constant probability of increasing the size of $\ME{i}$ by one and in a vertex-step we have a probability $\approx \alpha/n$ of matching the vertex $v$ and thus no longer  
	including it in $\GE{i+1}$ and $\High{i+1}$. We formalize this in the following. 
	
	\begin{claim}\label{clm:ms-1st-matching-step}
		Fix $V_i$ and suppose step $i$ is a \underline{matching-step} and the event of~\Cref{clm:ms-1st-v} has happened. Then, 
		\[
			\Pr_{e_i}\paren{e_i \in \ME{i+1} \mid V_i} \geq \frac{1}{3}. 
		\] 
	\end{claim}
	\begin{proof}
		As $N(V_i)$ contains more than $(n/2\alpha)$ vertices (as this is matching-step) while $\Measy$ has at most $(n/4\alpha)$ vertices (as $\eventmatch(i)$ has happened), we know that at least half the vertices in $N(V_i)$ are unmatched. Given that all of $V_i$ is also unmatched, 
		if $\NES(G,V_i)$ samples $e_i$ to any of the unmatched vertices in $N(V_i)$, we can include $e_i$ in $\ME{i+1}$ greedily. As the choice of $\NES(G,V_i)$ is uniform over $N(V_i)$, this event 
		happens with probability at least $(1/2 - \delta_F) > 1/3$, as desired (since $\delta_F=1/100$).  \Qed{\Cref{clm:ms-1st-matching-step}}
		
	\end{proof}
	
	\begin{claim}\label{clm:ms-1st-vertex-step}
		Fix $V_i$ and suppose step $i$ is a \underline{vertex-step} and the event of~\Cref{clm:ms-1st-v} has happened. Then, 
		\[
			\Pr_{e_i}\paren{v \in V(\ME{i+1}) \mid V_i} \geq \frac{\alpha}{n}. 
		\] 
	\end{claim}
	\begin{proof}
		We know $v \in N(V_i)$ and that size of $N(V_i)$ is at most $(n/2\alpha)$. At the same time, since $V_i$ is clean, if $v$ is sampled as an endpoint of $e_i$ by $\NES(G,V_i)$, the edge $e_i$ will join the matching greedily and thus 
		$v$ will be matched. As the choice of $\NES(G,V_i)$ is uniform over $N(V_i)$ and $\delta_F =1/100$, 
		\[
			\Pr_{e_i}\paren{v \in V(\ME{i+1}) \mid V_i} \geq (1-\delta_F) \cdot \frac{1}{\card{N(V_i)}}  > \frac{\alpha}{n}. \Qed{\Cref{clm:ms-1st-vertex-step}}
		\]
		
	\end{proof}

	We can now conclude the proof of~\Cref{lem:ms-1st-batch} as follows. We have that at least half the steps are matching-steps or half of them are vertex-steps. We consider each case as follows. 
	
	\paragraph{When half the steps are matching-steps.} In this case, each matching-step $i$ increases size of $\ME{i}$ by one 
	with probability at least $(1/48)$ by~\Cref{clm:ms-1st-v,clm:ms-1st-matching-step}. Thus, 
	\[
		\Exp\card{\ME{s+1}} \geq (\frac s2) \cdot \frac{1}{48} = \frac{n^2}{2\alpha^2 \cdot \log^3{n}} \cdot \frac{1}{48} \gg n/\alpha,
	\]
	given that $\alpha \leq n^{1-\delta}$. Moreover, the distribution of $\ME{s+1}$ statistically dominates sum of $(s/2)$ Bernoulli random variables with mean $(1/48)$. 
	As such, by the Chernoff bound (\Cref{prop:chernoff}),
	\[
	\Pr\paren{\ME{s+1} < (n/8\alpha)} < \exp\paren{-n/\alpha} \ll 1/\poly{(n)},
	\]
	as $\alpha \leq n^{1-\delta}$. This implies that $\eventmatch(s+1)$ happens, proving~\Cref{lem:ms-1st-batch} in this case. 
	
	\paragraph{When half the steps are vertex-steps.} In this case, each vertex-step $i$ can independently match the vertex $v$ with probability at least $(\alpha/16\,n)$ by~\Cref{clm:ms-1st-v,clm:ms-1st-vertex-step}. Thus, 
	\[
		\Pr\paren{v \in \High{s+1}} \leq (1-\frac{\alpha}{16\, n})^{s/2} \leq \exp\paren{-\frac{\alpha}{16\,n} \cdot  \frac{n^2}{2 \cdot \alpha^2 \cdot \log^3{n}} } < \exp\paren{-\frac{n^{\delta/2}}{32}} \ll 1/\poly{(n)},
	\]
	where we use $\alpha \leq n^{1-\delta}$. 
	Thus, with high probability $v$ will not be part of $\High{s+1}$. A union bound over all the vertices $v \in V$ then ensures that $\High{s+1}$ will be empty with high probability, thus proving~\Cref{lem:ms-1st-batch} in this case too.
	
	\smallskip
	
	\emph{Remark:} We note that the definition of matching-steps and vertex-steps are 
	tailored to individual vertices in $V$; however, even if one vertex leads to having at least half of the steps as matching-steps, we can apply the argument of first part and conclude the proof. Thus, when applying the second part of the argument, we can assume that all vertices lead to half of the steps being vertex-steps, and so we can union bound over all of them.

\subsection{Second Batch: Proof of \texorpdfstring{\Cref{lem:ms-2nd-batch}}{Lemma} } \label{subsec:2nd-batch}

We now prove~\Cref{lem:ms-2nd-batch}. In the following, we condition on the event that $\High{s+1}$ (and $\High{i}$ for every $i \in (s,2s]$) is empty. 
Our goal is then to prove that at some step $i \in (s,2s]$, one of the events $\eventmatch(i)$ or $\eventsparsify(i)$ is not going to happen. 
The key to the proof of~\Cref{lem:ms-2nd-batch} (and~\Cref{lem:ms} itself) is the following. 

\begin{lemma}\label{lem:ms-increase}
	For any $i \in (s\,,\,2s]$, 
	\[
		\Pr_{(V_i,e_i)}\Paren{\ME{i+1} > \ME{i}} \geq \frac{\alpha \cdot \log^3\!{n}}{4 \cdot n}.
	\] 
\end{lemma}
 
We first identify a simple structure in the graph $\GE{i}$. The following claim is based on a standard low-degree orientation of the graph plus geometric grouping 
of degrees of vertices.

\begin{claim}\label{clm:ms-geasy-large}
	\underline{At least one} of the following two conditions is true about $\GE{i}$: 
	\begin{enumerate}[label=$(\roman*)$]
		\item for some $d \in \left[\log^3\!{n},\dfrac{n}{8\alpha}\right)$, there are $\left( \dfrac{n \cdot \log^3{n}}{2 \alpha d} \right)$ vertices $v$ in $\Med{i}$ with $\degG{v}{i} \geq d$;  
		\item for some $d \in [1,\log^3\!{n})$, there are $\left(\dfrac{19\,n\cdot\log^3{n}}{2 \alpha d} \right)$ vertices in $\Low{i}$ with at least $d$ neighbors in $\Low{i}$. 
	\end{enumerate}
\end{claim}
\begin{proof}
	Given that $\High{i}$ is empty, any edge in $\GE{i}$ is either incident on $\Med{i}$ or is between two vertices in $\Low{i}$. Consequently, given that by $\eventsparsify(i)$, we have at least 
	$(20\,(n/\alpha)\cdot\log^4\!{n})$ edges in $\GE{i}$, there are either at least $(n \cdot \log^4\!{n}/\alpha)$ edges incident on $\Med{i}$ or $(19 \cdot n \cdot \log^4 (n)/\alpha)$ edges entirely inside $\Low{i}$.
	We prove that each case corresponds to one of the conditions in the claim.
	
	\paragraph{When $\geq ((n/\alpha) \cdot \log^4 {n})$ edges are incident on $\Med{i}$.} We partition vertices in $\Med{i}$ into sets $\set{D_j}$ where each $D_j$ contains  vertices $v$ with $\degG{v}{i} \in [2^j , 2^{j+1})$. 
	As such, 
	\begin{align*}
		\sum_{j} \card{D_j} \cdot 2^{j+1} \geq \text{\# edges incident on $\Med{i}$} \geq (n/\alpha) \cdot \log^4\!{n}.
	\end{align*}
	As there are at most $\log{n}$ choices for $j$ in the summation above,  we should have some $D_{j^*}$ with 
	\[
	\card{D_{j^*}} \geq \dfrac{(n/\alpha) \cdot \log^3\!{n}}{2^{j^*+1}}.
	\]
	Setting $d = 2^{j^*}$ and returning (a subset of) $D_{j^*}$ satisfies the bound in part $(i)$ of the claim: all vertices in $D_{j^*} \subseteq \Med{i}$ have $\degG{\cdot}{i}$ in $[\log^3\!{n},\dfrac{n}{8\alpha})$ by definition of $\Med{i}$, and we can pick
	a subset of $D_{j^*}$ with size prescribed by the claim as all vertices in $D_{j^*}$ have degree $d$ at least. 

	\paragraph{When $\geq (19\,(n/\alpha) \cdot \log^4\!{n})$ edges are entirely inside $\Low{i}$.} The argument is almost identical to the above part by counting the degree of vertices in $\Low{i}$ but 
	only in $\Low{i}$ (instead of all of $\degG{\cdot}{i}$ as in the previous part). We partition vertices in $\Low{i}$ into sets $\set{D_j}$ where each $D_j$ contains all vertices with number of neighbors in $\Low{i}$ in $[2^j,2^{j+1})$. 
	As such, 
	\begin{align*}
		\sum_{j} \card{D_j} \cdot 2^{j+1} \geq \text{\# edges entirely inside $\Low{i}$} \geq 19\,(n/\alpha) \cdot \log^4\!{n}.
	\end{align*}
	As there are at most $\log{n}$ choices for $j$ in the summation above, we should have some $D_{j^*}$ with 
	\[
	\card{D_{j^*}} \geq \dfrac{19\,(n/\alpha) \cdot \log^3\!{n}}{2^{j^*+1}}.
	\]
	Setting $d = 2^{j^*}$ and returning (a subset of) $D_{j^*}$ satisfies the bound in part $(ii)$ of the claim: all vertices in $D_{j^*} \subseteq \Low{i}$ have degree less than $(\log^3\!{n})$ by the definition of $\Low{i}$ 
	(even in $\GE{i}$ and so between $\Low{i}$ also) and 
	we can pick a subset of $D_{j^*}$ with the required size as vertices in $D_{j^*}$ have degree $d$ at least. \Qed{\Cref{clm:ms-geasy-large}}
	
\end{proof}

In the following, we refer to a step $i \in (s,2s]$ as a \textbf{$\bm{V_i}$-step} whenever case $(i)$ of~\Cref{clm:ms-geasy-large} happens and a \textbf{$\bm{N(V_i)}$-step} otherwise. We will show that: 
\begin{itemize}
	 \item In a \textbf{$\bm{V_i}$-step}, we have ``enough'' large degree vertices and even if we sample one of them in $V_i$ it will make the intersection of $N(V_i)$ and $\GE{i}$ large;
	\item In a \textbf{$\bm{N(V_i)}$-step}, we have ``so many'' low degree vertices in $\GE{i}$ that many of them will appear in $N(V_i)$ and thus there is a large intersection 
between $N(V_i)$ and $\GE{i}$ again. 
\end{itemize}
In each case, we can  finalize the proof by showing that having $N(V_i)$ intersect largely with $\GE{i}$ allows us to recover an edge $e_i$ via $\NES(G,V_i)$ 
that can increase size of $\ME{i}$ with sufficiently large probability.

\subsubsection*{Case $(i)$ of~\Cref{clm:ms-geasy-large}: \textbf{$\bm{V_i}$-steps}} 
Let 
\begin{align}
d \in [\log^3\!{n},\dfrac{n}{8\alpha}) \quad \text{and} \quad D \subseteq \Med{i} \quad \text{with} \quad \card{D} = \frac{n\cdot\log^3{n}}{2 \alpha d} \label{eq:ms-Vsteps-1}
\end{align}
be, respectively, the degree-parameter and corresponding set guaranteed by Case $(i)$ of~\Cref{clm:ms-geasy-large}. The following claim lower bounds the probability that $V_i$ is both clean and  samples a vertex from $D$. 

\begin{claim}\label{clm:ms-Vstep-1}
	$\Pr_{V_i}\paren{V_i \cap D \neq \emptyset~\textnormal{and $V_i$ is clean} } \geq \dfrac{\log^3\!{n}}{8d}$. 
\end{claim}
\begin{proof}
	We have, 
\begin{align*}
	\Pr\paren{V_i \cap D \neq \emptyset ~\textnormal{and $V_i$ is clean}} &\geq \sum_{v \in D} \Pr\paren{v \in V_i ~\textnormal{and $V_i$ is clean}} \, - \sum_{u \neq w \in D} \Pr\paren{u,w \in V_i} 
	\tag{by inclusion-exclusion principle and dropping the `intersection' from the second event} \\
	&\geq \card{D} \cdot \frac{3}{4k}  - \card{D}^2 \cdot \frac{1}{k^2} \tag{by~\Cref{clm:ms-helper-1} and as $h_i(\cdot)$ is a pair-wise independent hash function with range $[k]$} \\
	&= \frac{n\cdot\log^3{n}}{2 \alpha d} \cdot \frac{\alpha}{n} \cdot \paren{\frac{3}{4} - \frac{n\log^3{n} \cdot \alpha}{ 2 \alpha d \cdot n}} \tag{by the choice of $k=n/\alpha$ and size of $D$ in~\Cref{eq:ms-Vsteps-1}} \\
	&\geq \frac{\log^3{n}}{8d},
\end{align*}
as $d \geq \log^3\!{n}$ by~\Cref{eq:ms-Vsteps-1}.  \Qed{\Cref{clm:ms-Vstep-1}}

\end{proof}

Let us now condition on the choice of $V_i$ and assume the event of~\Cref{clm:ms-Vstep-1} happens. Given that any vertex in $D$ already has $d$ neighbors in $\GE{i}$, 
we have that $N(V_i) \cap \GE{i}$ has size at least $d$ in this case. On the other hand, $N(V_i)$ can have at most $(n/4\alpha)$ neighbors outside $\GE{i}$ by the bound on the total number of 
matched vertices by~$\eventmatch(i)$. As the choice of $e_i$ from $\NES(G,V_i)$ is uniform over $N(V_i)$, we have, 
\begin{align*}
	\Pr_{e_i}\paren{\text{$e_i$ is from $V_i$ to $N(V_i) \cap \GE{i}$} \mid V_i} &\geq (1-\delta_F) \cdot \frac{\card{N(V_i) \cap \GE{i}}}{\card{N(V_i)}} \\
	&\geq (1-\delta_F) \cdot \frac{d}{(n/4\alpha) + d} \\
	&\geq (1-\delta_F) \cdot \frac{d \cdot 8\alpha}{3n} \tag{as $d \leq (n/8\alpha)$ in~\Cref{eq:ms-Vsteps-1}} \\
	&\geq \frac{2d \cdot \alpha}{n},
\end{align*}
as $\delta_F < 1/4$. Given that all of $V_i$ is also unmatched (as $V_i$ is clean by conditioning on the event of~\Cref{clm:ms-Vstep-1}), 
we can include $e_i$ in $\ME{i+1}$ greedily whenever $e_i$ is between $V_i$ and $N(V_i) \cap \GE{i}$.

Consequently, combining the two events above, we have, 
\begin{align*}
	\Pr_{(V_i,e_i)}\Paren{\ME{i+1} > \ME{i}} &\geq \Pr_{V_i}\paren{V_i \cap D \neq \emptyset~\textnormal{and $V_i$ is clean} } \cdot \Pr_{e_i}\paren{\text{$e_i$ is from $V_i$ to $N(V_i) \cap \GE{i}$} \mid V_i} \\
	&\geq  \dfrac{\log^3\!{n}}{8d} \cdot \frac{2d \cdot \alpha}{n} = \frac{\alpha\log^3\!{n}}{4\,n}.
\end{align*}
This concludes the proof of~\Cref{lem:ms-increase} in this case.

\subsubsection*{Case $(ii)$ of~\Cref{clm:ms-geasy-large}: \textbf{$\bm{N(V_i)}$-steps}}

 Let 
\begin{align}
d \in [1,\log^3\!{n}) \quad \text{and} \quad D \subseteq \Low{i} \quad \text{with} \quad \card{D} = \frac{19\,(n/\alpha)\cdot\log^3{n}}{2d} \label{eq:ms-NVsteps-1}
\end{align}
be, respectively, the degree-parameter and corresponding set guaranteed by Case $(ii)$ of~\Cref{clm:ms-geasy-large}. 
For the rest of this analysis, we focus only on the subgraph of $\GE{i}$ induced on vertices of $\Low{i}$ and for each  $v \in D$, we pick exactly $d$ (arbitrary) neighbors from $\Low{i}$ and 
denote them by $\NL(v)$.
Our goal is to show that $N(V_i)$ and $\GE{i}$  intersect largely.
We will do so by counting the elements in $D$ that have neighbors in $V_i$. This works because $D \subseteq \GE{i}$ and having neighbors in $V_i$ means that the vertex itself is in $N(V_i)$.

For any vertex $v \in D$, define an indicator random variable $X_v \in \set{0,1}$ which is $1$ iff $\NL(v) \cap V_i \neq \emptyset$ (see \Cref{fig:MorS_a}). Notice that $X= \sum_{v \in D} X_v$ is a random variable that denotes the number of vertices $v$ in $D$ that have a neighbor in $\NL(v)$ that belongs to $V_i$.
Note that we do not consider all neighbors of $v$ in $\Low{i}$, only the ones in $NL(v)$; this is okay since we just need a lower bound on $\card{N(V_i) \cap \GE{i}}$.
It is easy to see that $X \leq \card{N(V_i) \cap \GE{i}}$ since $v$ contributes to $\card{N(V_i) \cap \GE{i}}$ if $X_v=1$ (see \Cref{fig:MorS_b}).
We first bound the probability of the event $\NL(v) \cap V_i \neq \emptyset$.

\begin{figure}[!ht]
	\centering
	\subcaptionbox{This figure shows the set of vertices $\Low{i}$ and its subset $D$. $V_i$ (in blue) is the set of sampled vertices in step $i$. For a vertex $u$ in $D$ we define a set of neighbors $\NL(u)$ which if intersects with $V_i$ then we have random variable $X_u=1$.  \label{fig:MorS_a}}%
	[.45\linewidth]{ 	\resizebox{160pt}{110pt}{

\tikzset{every picture/.style={line width=0.75pt}} 

\begin{tikzpicture}[x=0.75pt,y=0.75pt,yscale=-1,xscale=1]

\draw   (217,17) -- (408,17) -- (408,187) -- (217,187) -- cycle ;
\draw  [color={rgb, 255:red, 25; green, 2; blue, 208 }  ,draw opacity=1 ] (266,156.5) .. controls (266,125.85) and (313.23,101) .. (371.5,101) .. controls (429.77,101) and (477,125.85) .. (477,156.5) .. controls (477,187.15) and (429.77,212) .. (371.5,212) .. controls (313.23,212) and (266,187.15) .. (266,156.5) -- cycle ;
\draw   (231,86) .. controls (231,66.12) and (247.12,50) .. (267,50) .. controls (286.88,50) and (303,66.12) .. (303,86) .. controls (303,105.88) and (286.88,122) .. (267,122) .. controls (247.12,122) and (231,105.88) .. (231,86) -- cycle ;
\draw   (283,98) .. controls (283,95.51) and (285.01,93.5) .. (287.5,93.5) .. controls (289.99,93.5) and (292,95.51) .. (292,98) .. controls (292,100.49) and (289.99,102.5) .. (287.5,102.5) .. controls (285.01,102.5) and (283,100.49) .. (283,98) -- cycle ;
\draw   (306.5,98) .. controls (306.5,85.3) and (312.66,75) .. (320.25,75) .. controls (327.84,75) and (334,85.3) .. (334,98) .. controls (334,110.7) and (327.84,121) .. (320.25,121) .. controls (312.66,121) and (306.5,110.7) .. (306.5,98) -- cycle ;
\draw    (287.5,102.5) -- (320.25,121) ;
\draw    (287.5,93.5) -- (320.25,75) ;

\draw (274,82) node [anchor=north west][inner sep=0.75pt]   [align=left] {$u$};
\draw (337,71) node [anchor=north west][inner sep=0.75pt]   [align=left] {$\NL(u)$};
\draw (477,169) node [anchor=north west][inner sep=0.75pt]   [font=\large,color={rgb, 255:red, 25; green, 2; blue, 208 }  ,opacity=1 ] [align=left] {$V_i$};
\draw (415,25) node [anchor=north west][inner sep=0.75pt]   [align=left] {$\Low{i}$};
\draw (235,40) node [anchor=north west][inner sep=0.75pt]   [align=left] {$D$};

\end{tikzpicture}
        }}
	\hspace{0.4cm}
	\subcaptionbox{This figure shows $X \leq \card{N(V_i) \cap \GE{i}}$. The vertices $u_j$ are in $\Low{i}$. Notice that $X_{u_1}=1,X_{u_2}=1,X_{u_3}=0$ and $X_{u_4}$ is not defined so we have $X=2$. But $\card{N(V_i) \cap \GE{i}}=~3$ since $u_1,u_2$ and $u_4$ contribute to it.  \label{fig:MorS_b}}%
	[.45\linewidth]{ \resizebox{200pt}{100pt}{

\tikzset{every picture/.style={line width=0.75pt}} 

\begin{tikzpicture}[x=0.75pt,y=0.75pt,yscale=-1,xscale=1]

\draw   (153,20) -- (441,20) -- (441,190) -- (153,190) -- cycle ;
\draw [color={rgb, 255:red, 25; green, 2; blue, 208 }  ,draw opacity=1 ]   (265,131.5) .. controls (265,100.85) and (312.23,76) .. (370.5,76) .. controls (428.77,76) and (476,100.85) .. (476,131.5) .. controls (476,162.15) and (428.77,187) .. (370.5,187) .. controls (312.23,187) and (265,162.15) .. (265,131.5) -- cycle ;
\draw   (206.5,82.25) .. controls (206.5,55.6) and (228.1,34) .. (254.75,34) .. controls (281.4,34) and (303,55.6) .. (303,82.25) .. controls (303,108.9) and (281.4,130.5) .. (254.75,130.5) .. controls (228.1,130.5) and (206.5,108.9) .. (206.5,82.25) -- cycle ;
\draw   (283,71) .. controls (283,68.51) and (285.01,66.5) .. (287.5,66.5) .. controls (289.99,66.5) and (292,68.51) .. (292,71) .. controls (292,73.49) and (289.99,75.5) .. (287.5,75.5) .. controls (285.01,75.5) and (283,73.49) .. (283,71) -- cycle ;
\draw   (306.5,71) .. controls (306.5,58.3) and (312.66,48) .. (320.25,48) .. controls (327.84,48) and (334,58.3) .. (334,71) .. controls (334,83.7) and (327.84,94) .. (320.25,94) .. controls (312.66,94) and (306.5,83.7) .. (306.5,71) -- cycle ;
\draw    (287.5,75.5) -- (320.25,94) ;
\draw    (287.5,66.5) -- (320.25,48) ;
\draw   (248,102) .. controls (248,99.51) and (250.01,97.5) .. (252.5,97.5) .. controls (254.99,97.5) and (257,99.51) .. (257,102) .. controls (257,104.49) and (254.99,106.5) .. (252.5,106.5) .. controls (250.01,106.5) and (248,104.49) .. (248,102) -- cycle ;
\draw   (276.41,98.54) .. controls (281.61,93.04) and (288.64,91.24) .. (292.12,94.53) .. controls (295.6,97.81) and (294.2,104.94) .. (289,110.44) .. controls (283.8,115.94) and (276.77,117.73) .. (273.3,114.45) .. controls (269.82,111.16) and (271.21,104.04) .. (276.41,98.54) -- cycle ;
\draw    (252.5,106.5) -- (273.3,114.45) ;
\draw    (252.5,97.5) -- (288,91.5) ;
\draw   (222,75) .. controls (222,72.51) and (224.01,70.5) .. (226.5,70.5) .. controls (228.99,70.5) and (231,72.51) .. (231,75) .. controls (231,77.49) and (228.99,79.5) .. (226.5,79.5) .. controls (224.01,79.5) and (222,77.49) .. (222,75) -- cycle ;
\draw   (167.5,83) .. controls (167.5,70.3) and (173.66,60) .. (181.25,60) .. controls (188.84,60) and (195,70.3) .. (195,83) .. controls (195,95.7) and (188.84,106) .. (181.25,106) .. controls (173.66,106) and (167.5,95.7) .. (167.5,83) -- cycle ;
\draw    (226.5,79.5) -- (181.25,106) ;
\draw    (226.5,70.5) -- (181.25,60) ;
\draw   (392.06,45.46) .. controls (394.55,45.5) and (396.54,47.54) .. (396.5,50.02) .. controls (396.47,52.51) and (394.43,54.5) .. (391.94,54.46) .. controls (389.46,54.43) and (387.47,52.39) .. (387.5,49.9) .. controls (387.54,47.42) and (389.58,45.43) .. (392.06,45.46) -- cycle ;
\draw   (391.75,68.96) .. controls (404.45,69.13) and (414.67,75.42) .. (414.57,83.02) .. controls (414.47,90.61) and (404.09,96.63) .. (391.39,96.46) .. controls (378.69,96.29) and (368.47,90) .. (368.57,82.41) .. controls (368.67,74.81) and (379.05,68.79) .. (391.75,68.96) -- cycle ;
\draw    (387.5,49.9) -- (368.57,82.41) ;
\draw    (396.5,50.02) -- (414.57,83.02) ;

\draw (269,48) node [anchor=north west][inner sep=0.75pt]   [align=left] {$u_2$};
\draw (474,150) node [anchor=north west][inner sep=0.75pt]  [font=\large,color={rgb, 255:red, 25; green, 2; blue, 208 }  ,opacity=1 ]  [align=left] {$V_i$};
\draw (451,22) node [anchor=north west][inner sep=0.75pt]   [align=left] {$\Low{i}$};
\draw (284,23) node [anchor=north west][inner sep=0.75pt]   [align=left] {$D$};
\draw (236,109) node [anchor=north west][inner sep=0.75pt]   [align=left] {$u_1$};
\draw (221,50) node [anchor=north west][inner sep=0.75pt]   [align=left] {$u_3$};
\draw (403,33) node [anchor=north west][inner sep=0.75pt]   [align=left] {$u_4$};

\end{tikzpicture}
        }}
	\caption{Illustration of random variables $X_v$.}
	\label{fig:MorS}
\end{figure}

\begin{claim}\label{clm:ms-NVsteps-1}
	For any $v \in D$, 
	\[
	(1-o(1)) \cdot \frac{d \cdot \alpha}{n} \leq \Pr\paren{X_v = 1} \leq \frac{d \cdot \alpha}{n}.
	\]
\end{claim}
\begin{proof}
	$X_v = 1$ iff one of the neighbors of $v$ in $\NL(v)$ belongs to $V_i$. For the upper bound, by union bound, 
	\begin{align*}
		\Pr\paren{X_v = 1} \leq \sum_{u \in \NL(v)} \Pr\paren{u \in V_i} = \card{\NL(v)} \cdot \frac{1}{k} = \frac{d \cdot \alpha}{n}. \tag{as $\card{\NL(v)} = d$ and $k = n/\alpha$}
	\end{align*}
	For the lower bound, by inclusion-exclusion principle, 
	\begin{align*}
		\Pr\paren{X_v = 1} &\geq \sum_{u \in \NL(v)} \Pr\paren{u \in V_i}  - \sum_{u \neq w \in \NL(v)} \Pr\paren{u,w \in V_i} \\
		&\geq \card{\NL(v)} \cdot \frac{1}{k} - \card{\NL(v)}^2 \cdot \frac{1}{k^2} \\
		&= \frac{d \cdot \alpha}{n} \cdot \paren{1-\frac{d \cdot \alpha}{n}} \tag{as $\card{\NL(v)} = d$ and $k = n/\alpha$} \\
		&\geq (1-o(1)) \cdot \frac{d \cdot \alpha}{n},
	\end{align*}
	as $d < \log^3{n}$ by~\Cref{eq:ms-NVsteps-1} and $\alpha \leq n^{1-\delta}$. \Qed{\Cref{clm:ms-NVsteps-1}}
	
\end{proof}

By~\Cref{clm:ms-NVsteps-1} and the size of $D$ in~\Cref{eq:ms-NVsteps-1}, we have, 
\begin{align}
	(1-o(1)) \cdot \frac{19}{2}  \cdot \log^3{n} \leq \expect{X} \leq  \frac{19}{2} \cdot \log^3{n}. \label{eq:ms-NVsteps-exp}
\end{align}
Our goal now is to prove that $X$ is  concentrated. This requires a non-trivial proof as the variables $\set{X_v}_{v \in D}$ are correlated through their shared neighbors in $V_i$. But the fact that the subgraph 
induced on $\Low{i}$ is low-degree allows us to bound the variance of $X$ using a combinatorial argument in the following claim.  

\begin{claim}\label{clm:ms-NVsteps-var}
	$\var{X} \leq (1/8) \cdot \expect{X}^2$. 
\end{claim}
\begin{proof}
For any two vertices $u \neq v \in D$, define $\Com(u,v) := \NL(u) \cap \NL(v)$ as the set of common
neighbors of $u$ and $v$ in subgraph of $\Low{i}$ defined by $\NL(\cdot)$ and let $\com(u,v) = \card{\Com(u,v)}$. We have, 
\begin{align}
	\var{X} &= \sum_{v \in D} \var{X_v} + \sum_{u \neq v \in D} \cov{X_u,X_v} \leq \expect{X} + \sum_{u \neq v \in D} \cov{X_u,X_v},
	\label{eq:ms-NVsteps-var}
\end{align}
as $X_v$ is an indicator random variable and thus $\var{X_v} \leq \expect{X_v}$. 
We thus need to bound the covariance-terms only. Recall that 
\begin{align*}
\cov{X_u,X_v} &= \expect{X_u \cdot X_v} - \expect{X_u}\expect{X_v} \\
&= \Pr\paren{\NL(u) \cap V_i \neq \emptyset \wedge \NL(v) \cap V_i \neq \emptyset} - \Pr\paren{X_u=1} \cdot \Pr\paren{X_v=1}.
\end{align*} 
We can bound the second part using~\Cref{clm:ms-NVsteps-1} for each probability-term. For the first part, notice that  for
$\NL(u) \cap V_i \neq \emptyset$ and $\NL(v) \cap V_i \neq \emptyset$ one of the following two things should happen: at least one of the shared neighbors of $u,v$ in $\Com(u,v)$ is chosen in $V_i$ or each of them separately have 
a neighbor in $\NL(u) - \Com(u,v)$ and $\NL(v) - \Com(u,v)$ those join $V_i$ (as $h_i(\cdot)$ is a pair-wise independent hash function, the probability of these two distinct vertices joining $V_i$ is independent). Thus, 
\begin{align*}
	\Pr\paren{\NL(u) \cap V_i \neq \emptyset \wedge \NL(v) \cap V_i \neq \emptyset} &\leq \sum_{w \in \Com(u,v)} \hspace{-0.5cm}\Pr\paren{w \in V_i} + 
	{\hspace{-0.5cm} \sum_{\substack{z_u \in \NL(u) - \Com(u,v)\\ z_v \in \NL(v) - \Com(u,v)}} \hspace{-1cm} \Pr\paren{z_u \in V_i} \cdot \Pr\paren{z_v \in V_i}} \\
	&\leq \com(u,v) \cdot \frac{1}{k} + d^2 \cdot \frac{1}{k^2} \tag{as $h_i(\cdot)$ is uniform over $[k]$ and $u,v \in D$ and each vertex in $D$ has exactly $d$ neighbors in $\NL(\cdot)$} \\
	&= \com(u,v) \cdot \frac{\alpha}{n} + \frac{d^2 \cdot \alpha^2}{n^2} \tag{as $k = n/\alpha$}. 
\end{align*}
Plugging in this for the first term of covariance and the bounds in~\Cref{clm:ms-NVsteps-1} for the second terms, we have, 
\[
	\cov{X_u,X_v} \leq \com(u,v) \cdot \frac{\alpha}{n} + \frac{d^2 \cdot \alpha^2}{n^2} - (1-o(1)) \cdot \frac{d^2 \cdot \alpha^2}{n^2} = \com(u,v) \cdot \frac{\alpha}{n} + o(1) \cdot \frac{d^2 \cdot \alpha^2}{n^2}. 
\]
By plugging in further in the RHS of~\Cref{eq:ms-NVsteps-var}, we get that, 
\begin{align*}
	\var{X} &\leq \expect{X} + \card{D^2} \cdot o(1) \cdot \frac{d^2 \cdot \alpha^2}{n^2} +  \sum_{u \neq v \in D} \com(u,v) \cdot \frac{\alpha}{n} \\
	&\leq \expect{X} + \paren{\frac{19 \cdot (n/\alpha) \cdot \log^3{n}}{2d}}^2 \cdot o(1) \cdot \frac{d^2 \cdot \alpha^2}{n^2} + \sum_{u \neq v \in D} \com(u,v) \cdot \frac{\alpha}{n} \tag{by the bound on size of $D$ in~\Cref{eq:ms-NVsteps-1}} \\
	&\leq o(1) \cdot \expect{X}^2 +  \frac{\alpha}{n} \cdot \sum_{u \neq v \in D} \com(u,v). \tag{by the lower bound on $\expect{X} \geq (1-o(1)) \cdot (19/2) \cdot \log^3{n}$ in~\Cref{eq:ms-NVsteps-exp}}
\end{align*}
The remaining part is then to compute the summation in the RHS which we do below using a double-counting argument. 
Note that $\sum_{u \neq v \in D} \com(u,v)$ counts the number of common neighbors inside $\NL(\cdot)$-subgraph of $\Low{i}$ for each pair of vertices in $D$. This can be alternatively 
counted by going over vertices in $\Low{i}$ that are neighbor to $D$ and count the number of pairs of neighbors (in $\NL(\cdot)$) they have in $D$. 
\begin{align*}
	\sum_{u \neq v \in D} \com(u,v) &= \sum_{z \in \NL(D)} {\binom{\card{\NL(z)}}{2}} \leq \sum_{z \in \NL(D)} \card{\NL(z)}^2 \\
	&\leq (\log^3{n}) \cdot \sum_{z \in \NL(D)} \card{\NL(z)} \tag{each vertex in $\Low{i}$ has degree at most $(\log^3{n})$ to $\Low{i}$ and $z \in \Low{i}$ as it is in $\NL(D)$} \\
	&\leq (\log^3{n}) \cdot \card{D} \cdot d \tag{as the sum-term counts the number of edges between $D$ and $\NL(D)$ which is at most $\card{D} \cdot d$}\\
	&\leq (\log^3{n}) \cdot (\frac{19}{2} \cdot (n/\alpha) \cdot \log^3{n}). \tag{by~\Cref{eq:ms-NVsteps-1} on the size of $D$}
\end{align*}

Plugging in this bound in the upper bound on $\var{X}$ in the earlier equation, we have, 
\begin{align*}
	\var{X} &\leq o(1) \cdot \expect{X}^2 +  \frac{\alpha}{n} \cdot \paren{\frac{19}{2} \cdot (n/\alpha) \cdot \log^6{n}} \\
	&= o(1) \cdot \expect{X}^2 +  \frac{19}{2} \cdot \log^6{n} \\
	&< \frac{1}{8} \cdot \expect{X}^2,  
\end{align*}
as $\expect{X} \geq  (1-o(1)) \cdot \dfrac{19}{2} \cdot (\log^3{n})$ by~\Cref{eq:ms-NVsteps-exp}.  \Qed{\Cref{clm:ms-NVsteps-var}} 

\end{proof}

Recall that 
\[
\card{N(V_i) \cap \GE{i}} \geq X.
\]
Given the bound on expectation and variance of $X$ in~\Cref{eq:ms-NVsteps-exp} and~\Cref{clm:ms-NVsteps-var}, respectively, we can now 
apply Chebyshev's inequality and get that,
\begin{align*}
	\Pr\paren{\card{N(V_i) \cap \GE{i}} < 2 \cdot \log^3{n}} \leq \Pr\paren{\card{X-\expect{X}} \geq \frac{15}{19} \cdot \expect{X}} \leq \frac{19^2 \cdot \var{X}}{15^2\cdot\expect{X}^2} \leq \frac{19^2}{15^2 \cdot 8} < \frac{1}{4}. 
\end{align*}
Additionally, we also have that the probability that $V_i$ is not clean is at most, 
\[
	\Pr\paren{\text{$V_i$ is not clean}} \leq \sum_{u \in V(\ME{i})} \Pr\paren{u \in V_i} < \frac{n}{4\alpha} \cdot \frac{\alpha}{n} = \frac{1}{4}. 
\]
By a union bound on the two equations above, we have, 
\begin{align}
\Pr\paren{\card{N(V_i) \cap \GE{i}} \geq 2\log^3{n} ~\text{and $V_i$ is clean}} \geq \frac12. 
\label{eq:ms-NVsteps-cond}
\end{align}

The rest of the proof is similar to that of \textbf{$\bm{V_i}$-steps}. We condition the choice of $V_i$ and assume the event of~\Cref{eq:ms-NVsteps-cond} has happened. 
Thus, we have that both $V_i$ is clean and $N(V_i)$ has at least $2 \log^3{n}$ vertices in $\GE{i}$. 
Moreover, $N(V_i)$ can have at most $(n/4\alpha)$ neighbors outside $\GE{i}$ by the bound on the total number of 
matched vertices by~$\eventmatch(i)$. As the choice of $e_i$ from $\NES(G,V_i)$ is uniform over $N(V_i)$, we have, 
\begin{align*}
	\Pr_{e_i}\paren{\text{$e_i$ is from $V_i$ to $N(V_i) \cap \GE{i}$} \mid V_i} &\geq (1-\delta_F) \cdot \frac{\card{N(V_i) \cap \GE{i}}}{\card{N(V_i)}} \\
	&\geq (1-\delta_F) \cdot \frac{2 \log^3{n}}{(n/4\alpha) + 2 \log^3{n}} \\
	&> \frac{3 \cdot \alpha \cdot \log^3{n}}{n},
\end{align*}
as $\alpha \leq n^{1-\delta}$ and $\delta_F < 1/4$. Given that all of $V_i$ is also unmatched (as $V_i$ is clean), 
we can include $e_i$ in $\ME{i+1}$ greedily whenever the event of the LHS above  happens. 

Consequently, combining the two events above, we have, 
\begin{align*}
	\Pr_{(V_i,e_i)}\Paren{\ME{i+1} > \ME{i}} \geq & \Pr_{V_i}\paren{\card{N(V_i) \cap \GE{i}} \geq 2 \log^3{n} ~\text{and $V_i$ is clean}} \cdot \\ & \Pr_{e_i}\paren{\text{$e_i$ is from $V_i$ to $N(V_i) \cap \GE{i}$} \mid V_i} \\
	\geq & \frac{1}{2} \cdot \frac{3 \cdot \alpha \cdot \log^3{n}}{n} > \frac{\alpha\log^3\!{n}}{n}.
\end{align*}
This concludes the proof of~\Cref{lem:ms-increase} in this case also. 

\subsubsection*{Concluding the Proof of~\Cref{lem:ms-2nd-batch}}

 By~\Cref{lem:ms-increase}, assuming the events $\eventmatch(i),\eventsparsify(i)$ hold for every $i\in (s,2s]$, size of $\ME{2s}$ 
statistically dominates sum of $s$ independent Bernoulli random variables $\set{Z_i}_{i=1}^s$ with mean $({\alpha \cdot \log^3\!{n}}/4n)$ (RHS of~\Cref{lem:ms-increase}). 
Let $Z = \sum_{i=1}^s Z_i$. Thus, by the choice of $s$ in~\Cref{alg:ms},
\[
\expect{Z} = \frac{n^2}{\alpha^2 \cdot \log^3{n}} \cdot \frac{\alpha \cdot \log^3\!{n}}{4 \cdot n} = \frac{n}{4 \cdot \alpha},
\]
 and by the Chernoff bound (\Cref{prop:chernoff}), 
\begin{align*}
	\Pr\paren{Z< \frac{n}{8\alpha}} < \Pr\paren{Z < \frac{1}{2} \cdot \expect{Z}} \leq \exp\paren{-\frac{n}{12\alpha}} \ll 1/\poly{(n)}, 
\end{align*}
where the final bound is by $\alpha \leq n^{1-\delta}$. This means that as long as $\eventmatch(i),\eventsparsify(i)$ happen for all $i \in (s:2s]$, with high probability we are going 
to end up with a matching $\Measy$ of size at least $(n/8\alpha)$, which means $\eventmatch(2s)$ does not happen as desired.

This concludes the proof of~\Cref{lem:ms}.

\end{document}